\documentclass[final,11pt]{article}
\usepackage[colorlinks=true,linkcolor=blue,citecolor=magenta]{hyperref}
\usepackage{amssymb,amsmath,amsthm}
\usepackage{cite}
\usepackage[utf8]{inputenc}
\usepackage{enumerate}
\usepackage{multicol}
\usepackage{lmodern}
\usepackage{microtype}
\usepackage[conditional,light,first,bottomafter]{draftcopy}
\draftcopyName{DRAFT\space\today}{130}
\draftcopySetScale{65}
\usepackage[letterpaper,hmargin=2.5cm,vmargin=2.8cm]{geometry}
\usepackage{setspace}
\makeatletter
\renewcommand{\section}{\@startsection%
{section}%
{1}%
{0em}%
{1.7em}%
{1.2em}%
{\normalfont\large\centering\bfseries}}
\renewcommand{\@seccntformat}[1]%
{\csname the#1\endcsname.\hspace{0.5em}}
\makeatother

\numberwithin{equation}{section}
\newtheorem{theorem}{Theorem}[section]
\newtheorem{proposition}[theorem]{Proposition}
\newtheorem{lemma}[theorem]{Lemma}
\newtheorem{corollary}[theorem]{Corollary}
\theoremstyle{definition}
\newtheorem{definition}[theorem]{Definition}
\newtheorem{remark}[theorem]{Remark}

\newcommand{\abs}[1]{\left|#1\right|}

\newcommand{\inner}[2]{\left\langle#1,#2\right\rangle}
\newcommand{\cc}[1]{\overline{#1}}

\newcommand{\nats}{\mathbb{N}}

\newcommand{\bea}{\begin{eqnarray}}
\newcommand{\eea}{\end{eqnarray}}
\newcommand{\beao}{\begin{eqnarray*}}
\newcommand{\eeao}{\end{eqnarray*}}

\newcommand{\llb}{\left\lbrace}
\newcommand{\rrb}{\right\rbrace}
\newcommand\R{{\mathbb R}}

\newcommand\N{{\mathbb N}}

\newcommand\C{{\mathbb C}}
\newcommand\tE{{\mathbb T}}
\renewcommand{\H}{\mathcal{H}}
\newcommand{\T}{{T}}
\newcommand{\HH}{{\mathcal{H}\oplus\mathcal{H}}}
\newcommand\ip[2]{\langle {#1},{#2} \rangle}

\newcommand\ipa[2]{\left\langle {#1},{#2} \right\rangle}
\newcommand\no[1]{\| {#1} \|}

\newcommand\noa[1]{\left\| {#1} \right\|}

\newcommand\mm[1]{{(#1)^{\perp}\oplus(#1)^{\perp}}}

\newcommand\CA[2]{\pmb{\mathsf{ C}}_{#1}({#2})}
\newcommand\CZ[2]{\pmb{\mathsf{ Z}}_{#1}({#2})}

\newcommand\pE[1]{{#1}_{\tiny{\odot}}}
\newcommand\Nk[2]{\pmb{\mathsf{N}}_{#1}({#2})}

\newcommand\oP[2]{{#1}\oplus {#2} }
\newcommand\oM[2]{{#1}\ominus {#2} }
\newcommand\cA[1]{\mathcal {#1}}
\newcommand\mF[1]{\pmb{\mathsf{#1}}}

\newcommand\rE[1]{_{\upharpoonright_{#1}}}
\newcommand\gA[2]{#1^\#(#2)}
\newcommand\vE[2]{{\begin{pmatrix}{#1}\\{#2}\end{pmatrix}}}
\DeclareMathOperator{\im}{Im}
\DeclareMathOperator{\re}{Re}
\DeclareMathOperator{\dom}{dom}
\DeclareMathOperator{\ran}{ran}
\DeclareMathOperator{\mul}{mul}

\DeclareMathOperator{\Span}{span}

\DeclareMathOperator{\Assoc}{Assoc}

\makeatletter
\newcommand{\mathleft}{\@fleqntrue\@mathmargin0pt}
\newcommand{\mathcenter}{\@fleqnfalse}
\makeatother

\begin{document}
\begin{titlepage}
\title{Dissipative extension theory for linear relations
\footnotetext{%
Mathematics Subject Classification(2010):
47A06  
47A45  
47B44  
}
\footnotetext{%
Keywords:
Closed linear relations;
Dissipative extensions;
Nondensely defined operators.
}
\hspace{-8mm}
\thanks{%
Research partially supported by SEP-CONACYT CB-2015 254062 and
UNAM-DGAPA-PAPIIT IN110818
}%
}
\author{
\textbf{Josu\'e I. Rios-Cangas}
\\
Departamento de F\'{i}sica Matem\'{a}tica\\
Instituto de Investigaciones en Matem\'aticas Aplicadas y en Sistemas\\
Universidad Nacional Aut\'onoma de M\'exico\\
C.P. 04510, Ciudad de M\'exico\\
\texttt{jottsmok@gmail.com}
\\[4mm]
\textbf{Luis O. Silva}
\\
Departamento de F\'{i}sica Matem\'{a}tica\\
Instituto de Investigaciones en Matem\'aticas Aplicadas y en Sistemas\\
Universidad Nacional Aut\'onoma de M\'exico\\
C.P. 04510, Ciudad de M\'exico\\
\texttt{silva@iimas.unam.mx}
}
\date{}
\maketitle
\vspace{-4mm}
\begin{center}
\begin{minipage}{5in}
  \centerline{{\bf Abstract}} \bigskip
\large
 This work is devoted to
  dissipative extension theory for dissipative linear relations. We
  give a self-consistent theory of extensions by generalizing the
  theory on symmetric extensions of symmetric operators. Several
  results on the properties of dissipative relations are
  proven. Finally, we deal with the spectral properties of dissipative
  extensions of dissipative relations and provide results concerning
  particular realizations of this general setting.
\normalsize
\end{minipage}
\end{center}
\thispagestyle{empty}
\end{titlepage}

\section{Introduction}
\label{sec:intro}

This paper deals with the theory of dissipative extensions of
dissipative relations and it can be seen as a generalization of the
classical von Neumann theory of symmetric extensions of symmetric
operators \cite{MR1512569}.  The theory is presented thoroughly and the exposition
goes along the lines of the classical texts on the von Neumann theory
(see for instance \cite[Chap.\, 7]{MR1255973},
\cite[Chap.\,4]{MR1192782} \cite[Chap.\,8]{MR566954}), but in a more
general setting.

In this work we obtain several results in the theory of dissipative
relations. Some of them can surely be considered mathematical folklore
for which, to the best of our knowledge, there were no proofs in the
literature prior to this work. It is also worth remarking that the
proofs of some classical results on dissipative operators, as well as
the ones on the particular instances of symmetric and selfadjoint
operators, are simplified and streamlined when these proofs are
considered in the more general framework of dissipative relations.

Our motivation for studying relations and its extensions comes from
their use in the boundary triplet theory
\cite{MR1087947,MR1318517,MR2486805,MR1154792} and quasi boundary
triplet theory \cite{MR2289696,MR3050306} for extensions of
symmetric operators; a panoramic account on boundary triplets is in
\cite[Chap.\,14]{MR2953553}. The theory of relations is also used in
studying extensions of nondensely defined symmetric operators (see for
instance \cite{MR0322568} and cf. \cite{MR0024061}). We remark that
the examples given in Section~\ref{sec:applications} are related to
this kind of applications. Relations are also relevant in other
contexts; for instance in the theory of canonical systems (see
\cite{MR1759823,MR1761504}).

It is not a coincidence that von Neumann was not only the pioneer in
extension theory of operators, but also in the theory of linear
relations. Indeed, the modern notion of linear relation goes back to
\cite{MR0034514}. The theory was later developed in
\cite{MR0123188,MR0477855,MR0361889}. More recent accounts on
the matter can be found in \cite{MR1631548,MR2327982}.
Symmetric extension theory of symmetric
relations was first developed in \cite{MR0361889} (cf. \cite{MR3178945,MR2093073}).
Various aspects of the theory of symmetric relations were studied in
\cite{MR0463964,MR1631548}. The perturbation theory of linear
relations is dealt with in \cite{MR3052575,MR1430397,MR3255523,MR2481074}.

The theory of dissipative operators has its roots in the theory of
contractions for which a seminal work is Sz. Nagy's
\cite{MR0058128}. Contractive and dissipative operators are related
via the Cayley transform (see \cite[Chap.\,4, Sec.\,4]{MR2760647}).
One of the first works on dissipative operators is due to Philips
\cite{MR0104919}. The development of Sz. Nagy and Foia\c{s}'s theory
for dissipative operators was done in \cite{MR0385642,MR0365199} and
later generalized in \cite{MR573902,MR513172,MR0500225}. Dissipative
extension theory was formulated in \cite{MR0104919}.

The theory presented here generalizes previous results in two
directions. We consider relations which are dissipative extensions of
dissipative relations. This general setting not only covers all
earlier results, but also shed light on the peculiarities of
dissipative relations that may be important in further developments
and applications (as for instance in the context of boundary triplets
for partial differential equations where the deficiency indices are
infinite). Dissipative relations appear in applications in
\cite{MR1318517,MR2486805} and are studied in
\cite{MR0361889,MR3057107}.

The paper is organized as follows. In
Section~\ref{sec:linear-relations}, we give a general account on the
theory of closed linear relations. Here we lay out the notation and
introduce preparatory facts. Section~\ref{sec:dissipative-relations}
is concerned with the theory of dissipative relations. In this
section, we extend some results on the  characterization of dissipative
operators to the case of relations (Theorems~\ref{codrih} and
\ref{inadmco}, and
Proposition~\ref{inadco}). Theorems~\ref{thm:sum-of-dissipative-dissipative}
and \ref{psewpos} give criteria for maximal dissipativeness for sums
of dissipative relations. Proposition
\ref{lsimresin} allows us to study the spectrum and the deficiency
index of a dissipative relation in terms of the spectrum and the
deficiency index of the operator part of it. In
Section~\ref{sec:dissipative-extensions}, we deal with dissipative
extensions of dissipative relations. Here, instead of using the Cayley
transform for relations (see \cite[Sec.\,2]{MR0361889}), we recur to
its modern counterpart, the $Z$-transform, introduced in
\cite{MR2093073}. Theorem~\ref{Fovon01d} provides the generalization
of the von Neumann formula for which Corollary~\ref{extresind0} and
Propositions~\ref{mdrofsrn} and \ref{mdrofsrn0} are related
results. The spectral properties of dissipative relations are dealt
with in Proposition~\ref{lainesp} and Corollary~\ref{lainesp0}.
Finally, Section~\ref{sec:applications} presents examples of
dissipative extensions for a Jacobi operator and the operator of
multiplication in a de Branges space in a general setting including
the case when they are not densely defined.

\section{Spectral theory of closed linear relations}
\label{sec:linear-relations}
Let $\H$
be a separable Hilbert space with inner product $\inner{\cdot}{\cdot}$
being antilinear in the first argument. Consider the orthogonal sum
of $\H$
with itself, i.\,e. $\HH$ (see \cite[Chap.\,2
Sec.\,3.3]{MR1192782}),
and denote an arbitrary element of it as a pair $\vE{f}{g}$
with $f,g\in\H$. Thus,
\begin{equation}\label{ipoHH}
\ipa{\vE{f_1}{g_1}}{\vE{f_2}{g_2}}=\ip{f_1}{f_2}+\ip{g_1}{g_2}.
\end{equation}
We shall use the norm
$$ \noa{\vE fg}=\no f+\no g,$$
which is equivalent to the norm
\bea\label{noaHH}
\noa{\vE fg}^{2}=\no f^{2}+\no g^{2}
\eea
generated by the inner product \eqref{ipoHH}.

Define the operators $\mF U,\ \mF W$ acting on $\HH$ by the rules
\begin{equation}
\label{eq:UW-properties}
\mF U\vE fg =\vE gf,\qquad
\mF W\vE fg =\vE{-g}f.
\end{equation}
One verifies that
$$\mF U^2=\mF I=-\mF W^2,\hskip3mm \mF U\mF W=-\mF W\mF U\,,$$
where $\mF I$ is the identity operator in $\HH$.
Moreover, for any linear subset $\cA G$ of $\HH$ the following holds
\begin{equation}
  \label{eq:UW-properties2}
  \begin{split}
    (\mF W\cA G)^{\perp}=\mF W(\cA G^{\perp})&\qquad\overline{\mF W\cA
      G}=\mF W\overline{\cA G}\\
(\mF U\cA G)^{\perp}=\mF U(\cA G^{\perp})&\qquad\overline{\mF U\cA G}=\mF U\overline{\cA G}
  \end{split}
\end{equation}

Throughout this paper, any linear set $\T$ in $\HH$ is called a linear
relation or simply a relation. The graph of a linear operator is a
relation, and thus any operator can be seen as a particular instance
of a relation. Not all relations are graphs of operators since for a
linear relation $\cA G$ to be the graph of an operator, it is
necessary and sufficient that
\begin{equation}
  \label{eq:cond-operator}
  \llb\vE f{g}\in \cA G\ : \ f=0\rrb=\llb\vE 00 \rrb.
\end{equation}

A closed relation is a subspace (closed linear set) in $\HH$. If a closed
relation is an operator, then the operator is closed \cite[Chap.\,3,
Sec.\,2]{MR1192782}.

For a given relation $\T$, define the sets
\begin{equation}
  \label{eq:dom-rank-ker-mul}
  \begin{split}
  \dom \T &:=\llb f\in \H\,:\ \vE fg\in T\rrb\quad
  \ran\T:=\llb g\in \H\,:\ \vE fg\in T\rrb\\
  \ker\T&:=\llb f\in \H\,:\ \vE f0\in T\rrb\quad
  \mul\T:=\llb g\in \H\,:\ \vE 0g\in T\rrb
  \end{split}
\end{equation}
which are linear sets in $\H$. Moreover, if $\T$
is closed, then $\ker\T$ and $\mul\T$ are subspaces of $\H$. According to
\eqref{eq:cond-operator} a relation is an operator if and only
if $\mul\T=\{0\}$.

Let $T$ and $S$ be relations, and $\zeta\in \C$. Consider the relations:
\begin{equation}
  \label{eq:sum-scalar-multiplication}
  \begin{split}
T+S&:=\llb\vE f{g+h}\ :\ \vE fg\in T,\ \ \vE fh\in S\rrb\quad
\zeta T:=\llb \vE f{\zeta g}\ :\ \vE fg\in T\rrb\\
ST&:=\llb \vE fk\ :\ \vE fg\in T,\ \ \vE gk\in S\rrb\qquad\T^{-1}:=\mF
U\T\,.
\end{split}
\end{equation}
Note that $\T^{-1}$ is the inverse of the relation $\T$. Clearly,
\begin{equation}
  \label{eq:inverse-sets-relations}
  \begin{split}
\dom \T^{-1}=\ran \T&\qquad\ran \T^{-1}=\dom \T\\ \ker \T^{-1}=\mul
\T &\qquad\mul\T^{-1}=\ker\T\\
(TS)^{-1}&=S^{-1}T^{-1}\,.
 \end{split}
\end{equation}
We also deal with the relations:
\begin{equation}
  \label{eq:sets-relations-operations}
  \begin{split}
T\dotplus S&:=\llb \vE {f+h}{g+k}\ :\ \vE fg\in T,\ \vE hk\in S,
\mbox{ and  }
T\cap S=\llb\vE00\rrb\rrb\,.\\
 \oP TS&:=T\dotplus S\,, \mbox{ with }\, T\subset S^\perp. \\
 \oM TS &:= T\cap S^\perp\,.
\end{split}
\end{equation}
Clearly, $T^\perp$ is a closed relation.
Note that, in the last two definitions, we consider the
orthogonal sum and difference in $\HH$. It will cause no confusion to
use the same symbol $\oplus$ when referring to subspaces of a Hilbert
space and when forming the orthogonal sum of Hilbert spaces.

Define
\begin{align*}
 \T^*:=\llb\vE hk\in \HH\ :\ \ip kf=\ip hg,\ \ \forall \vE fg\in \T\rrb.
\end{align*}
$\T^*$ is the adjoint of $T$ and has the following properties:
\begin{align}  \label{epoar}
\T^*&=(\mF W\T)^{\perp},&S\subset  T&\Rightarrow T^*\subset S^*,\nonumber\\
\T^{**}&=\overline T,& (\alpha\T)^*&=\overline{\alpha}\T^*,\,\mbox{ with } \alpha\neq0,\\
(T^*)^{-1}&=(T^{-1})^*,&\ker \T^*&=(\ran \T)^{\perp}.\nonumber
\end{align}
The last item above implies
\begin{align}\label{poHs}
\H=\oP{\overline{\ran \T}}{\ker \T^*}.
\end{align}
\begin{proposition}\label{arens}
If $T$ is a closed linear relation, then $\mul T=(\dom T^*)^\perp.$
\end{proposition}
\begin{proof}
It follows from \eqref{eq:inverse-sets-relations} and \eqref{epoar}
that
\begin{equation*}
  \mul T=\ker T^{-1}=\ker [({T^{*}})^{-1}]^{*}=[\ran ({T^{*}})^{-1} ]^{\perp}=(\dom T^*)^\perp\,.
\end{equation*}
\end{proof}

For the linear relations $S$ and $T,$ one directly verifies
\begin{equation}
  \label{eq:inclusion-sum-adjoints-adjoint-sum}
S^{*}+T^{*}\subset(S+T)^{*}\,.
\end{equation}
The conditions for the equality in the above inclusion are given by
the next assertion, which follows from the proof of \cite[Thm.\,3.41]{MR0123188}.
\begin{proposition}\label{linadj}
  If the domain of $\T$ is in the domain of $S$ and the domain of
  $(T+S)^*$ is in the domain of $S^*$, then $(S+\T)^ *=S^*+\T^*$.
\end{proposition}

We shall say that a relation $T$ is bounded if there exists $C>0$ such
that for all $\vE fg\in T$ it holds $\no g\leq C\,\no f$. It follows
from this definition that a bounded relation is a bounded operator.

\begin{remark}
  \label{rem:birman-stability-closedness}
  Repeating the proof of \cite[Thm.\,3.2.3]{MR1192782}, one verifies
  that if $T$ and $S$ are two closed relations such that $S$ is
  bounded, then $T+S$ is a closed relation.
\end{remark}

Define the quasi-regular set of $T$ by
\begin{equation*}
\hat\rho(T):=\{\zeta \in \C\ :\ (T-\zeta I)^{-1}\ \mbox{ is
  bounded}\}
\end{equation*}
As in the case of operators, the quasi-regular set is open.
It is a well-known fact, and a useful one, that a bounded
operator $T$ is closed if and only if its domain is closed.
\begin{proposition}\label{acot-cerr1}
  For every $\zeta\in \hat\rho(\T)$ it holds that $\ran (T-\zeta I)$
  is closed if and only if $T$ is closed.
\end{proposition}
\begin{proof}
  We suppose that $\ran (T-\zeta I)=\dom (T-\zeta I)^{-1}$ is closed,
  then $(T-\zeta I)^{-1}$ is closed, whence $T-\zeta I$ and $T$ are
  simultaneously closed (see
  Remark~\ref{rem:birman-stability-closedness}).  Conversely,
  closedness of $T$ implies both closedness of $(T-\zeta I)$ and
  $(T-\zeta I)^{-1}.$ Therefore $\dom (T-\zeta I)^{-1}$ is closed and
  by \eqref{eq:inverse-sets-relations} the assertion follows.
\end{proof}
Similar to what happens to operators, the deficiency index
\bea\label{ranTp} \eta_{\zeta}(T):=\dim \ran(T-\zeta I)^\perp \eea
is constant on each connected component of $ \hat\rho(T)$ when $T$ is closed
(cf. \cite[Chp.\,3 Sec.\,7 Lem.\,3]{MR1192782}).

For a linear relation $T$ and $\zeta\in \C$, we introduce the
deficiency space \bea\label{defspace} \Nk\zeta T:=\llb \vE f{\zeta
  f}\in T\rrb\,, \eea which is a closed bounded relation and $\dom {\Nk
  \zeta T}=\ker (T-\zeta I)$. Hence, it follows from \eqref{epoar} and
Proposition \ref{linadj} that, for $\zeta\in\hat\rho(T)$,
\beao \eta_{\zeta}(T)=\dim \Nk
{\overline\zeta }{T^{*}}.  \eeao
Both the deficiency index and the
deficiency space play a crucial role in the theory of dissipative
extensions of symmetric relations developed in Section
\ref{sec:dissipative-relations}.

If $\zeta\in\hat{\rho}(T)$ is such that $(T-\zeta I)^{-1}$ is in $\cA
B(\H)$ (the class of bounded operators defined on the whole
space $\H$), then $\zeta$ is in the regular set of $T$, which is denoted by
$\rho(T)$. When $T$ is closed, $\rho(T)$ is the union of the
connected components of $\hat{\rho}(T)$ in which
$\eta_{\zeta}(T)=0$. Note that whenever the relation is not closed, 
its regular set is empty.
\begin{proposition}
  \label{prop:resolvent-adjoint}
  Let $T$ be a closed linear relation. If $\zeta\in\rho(T)$, then
  $\cc{\zeta}\in\rho(T^*)$.
\end{proposition}
\begin{proof}
  The fact that $\zeta\in\rho(T)$ means that
  $(T-\zeta I)^{-1}\in \cA B(\H)$. This implies that
  $\left[(T-\zeta I)^{-1}\right]^*\in \cA B(\H)$ (see \cite[Chap.\,2
  Sec.\,4]{MR1192782}) which yields
  $(T^*-\cc{\zeta}I)^{-1}\in \cA B(\H)$ by \eqref{epoar}.
\end{proof}
In analogy to the operator case, the spectrum of the linear relation
$T$, denoted $\sigma(T)$, and its spectral core, $\hat\sigma(T)$, are
the complements in $\C$ of $\rho(T)$ and $\hat\rho(T)$, respectively.
As in the case of operators, one has
\begin{equation*}
  \hat\sigma(T)=\sigma_{p}(T)\cup\sigma_{c}(T)\,,
\end{equation*}
where the point spectrum, $\sigma_p(T)$, and the continuous spectrum,
$\sigma_c(T)$, are given by
\begin{equation*}
\begin{split}
\sigma_p(T)&:=\{\zeta \in \C\ :\ \ker (T-\zeta I)\neq \{0\}\}\quad\text{
    and }\\
\sigma_c(T)&:=\{\zeta \in \C\ :\ \ran (T-\zeta I)\neq \overline{\ran
  (T-\zeta I)}\}\,.
\end{split}
\end{equation*}

For a closed relation $T,$ define
\begin{equation*}
T_\infty:=\llb\vE 0g\in T\rrb,\quad\pE T:=\oM {T}{T_\infty}.
\end{equation*}
Thus,
\begin{equation}
  \label{eq:decomposition-relation}
  T=\oP {\pE T}{T_\infty}\,.
\end{equation}
Note that $\ran T_{\infty}=\mul T$ and $\pE T$ is a closed
operator. Moreover, $\dom \pE T$ coincides with $\dom T$ and $\pE T\subset T$. We say
that $\pE T$ is the operator part of $T$ and $T_{\infty}$ is the
multivalued part of $T$.

Apart from (\ref{eq:decomposition-relation}), there are alternative
decompositions of linear relations, not necessarily closed, into its
\emph{regular} and \emph{singular} parts \cite{MR2327982}.

The decomposition (\ref{eq:decomposition-relation}) allows us to study
some spectral properties of the relation $T$ by means of its operator
part. The next results deal with this matter.
\begin{proposition}\label{espoper00}
 If $T$ is a closed relation, then $\hat\rho(T)\subset\hat\rho(\pE T).$
\end{proposition}
\begin{proof}
  Observe that $(\pE T-\zeta I)^{-1}\subset( T-\zeta I)^{-1},$ for any
  $\zeta \in \C.$ If $\zeta\in \hat\rho(T)$, then $(T-\zeta
  I)^{-1}$ is bounded. Thus $(\pE T-\zeta I)^{-1}$ is bounded and
  $\zeta\in \hat\rho(\pE T).$
\end{proof}
The condition for the equality in the above result is given by the next assertion.
\begin{proposition}\label{espoper01}
  If $T$ is a closed relation such that $\dom T\subset(\mul
  T)^{\perp}$, then $\hat\rho(T)=\hat\rho(\pE T)$.
\end{proposition}
\begin{proof}
  It suffices to show that $(T-\zeta I)^{-1}$ is bounded when $\zeta
  \in \hat\rho(\pE T)$.

  If $\vE hk\in (T-\zeta I)^{-1}$, that is $\vE k{h+\zeta k}\in T$,
  then there exist $\vE{k}f\in \pE T$ and $\vE 0g\in T_\infty$ such
  that $\vE k{h+\zeta k}=\vE k{f+g}$. Thus \bea\label{eqaop0}
  h=f-\zeta k+g.  \eea Note that $\vE {f-\zeta k}k\in (\pE T-\zeta
  I)^{-1}$ and there exists $C>0$ such that \bea\label{eqaop1} \no
  k\leq C\,\no{f-\zeta k}.  \eea Since $f$ and $k$ are
  orthogonal to $g$, one has \bea\label{eqaop2} \no {f-\zeta
    k+g}^2=\no {f-\zeta k}^2+\no g^2.  \eea Combining
  \eqref{eqaop0}, \eqref{eqaop1}, and \eqref{eqaop2}, one obtains
\begin{align*}
 \no k^2&\leq C^2( \no{f-\zeta k}^2+\no g^2)\\
 &=C^2\no{f-\zeta k+g}^2\\&=C^2\no h^2.
\end{align*}
Therefore $\no k\leq C\,\no h$, which means that $(T-\zeta I)^{-1}$ is bounded.
\end{proof}

For the relations $T$ and $S$, define the relation $T_{S}$ in the
Hilbert space $\mm {\mul S}$ (with inner product inherited from $\HH$)
by \bea\label{redrel} T_{S}:=T\cap\mm{\mul S}.\eea The relation
$T_{S}$ is a linear relation.  Note that if $T$ is closed, then
$T_{S}$ is also closed and if $T$ is an operator, then $T_{S}$ is also
an operator in $(\mul S)^{\perp}$. Furthermore, one can verify that
$(T_{S})^{-1}=(T^{-1})_{S}$. In some cases, it is useful to consider
$T_{S}$ as a linear relation in $\HH$ and then $T_{S}\subset T$.
\begin{proposition}\label{opvsrc0} If $T$ is a closed relation, then
  $T_{T}=(\pE T)_{T}$ and, therefore, $T_{T}$ is a closed
  operator.\end{proposition}
\begin{proof}
  If $\vE {f}{h}\in T_{T}$, then, in view of \eqref
  {eq:decomposition-relation}, there are $\vE fg\in\pE T$, $\vE 0k \in
  T_{\infty}$ such that $h=g+k$. Since $h,g\in (\mul T)^ {\perp}$ and
  $k\in\mul T$, one has $h=g$. Hence $\vE fh\in(\pE T)_{T}$. The other
  inclusion follows directly by noting that $\pE T\subset T$.
\end{proof}
\begin{remark}\label{domrantop}
For a closed relation $T$ with $\dom T\subset(\mul T)^{\perp}$, it follows that $\dom \pE T$ and  $\ran \pE T$  are in  $(\mul T)^{\perp}.$
Thus by Proposition \ref{opvsrc0}
\bea\label{opars}
T_{T}=(\pE T)_{T}=\pE T\cap \mm{\mul T}=\pE T\,.
\eea
This means that $\pE T$  and $T_{T}$ have the same elements and, when
$T_T$ is regarded as a relation in $\HH$, one can write $$T=\oP{T_{T}}{T_{\infty}}\,.$$
Besides, for any $\zeta \in \C$,  
\begin{align}\label{depolr}
\begin{split}
T-\zeta I&=\oP{(\pE T-\zeta I)}{T_{\infty}}\\
&=\oP{( T_{T}-\zeta I)}{T_{\infty}}.
\end{split}
\end{align}
\end{remark}
\begin{theorem}\label{PopvsRel}
If $T$ is a closed linear relation such that $\dom T\subset(\mul T)^{\perp}$, then
 \begin{enumerate}[{(a)}]
\begin{multicols}{2}
   \item $\hat\sigma(T)=\hat\sigma(T_{T})$
  \item$\sigma(T)=\sigma(T_{T})$
  \item $\sigma_c(T)=\sigma_c(T_{T})$
    \item $\sigma_p(T)=\sigma_p(T_{T})$.
 \end{multicols}
 \end{enumerate}
 \end{theorem}
 \begin{proof}
 \begin{enumerate}[{(a)}]
\item  The subspaces $(T_{T}-\zeta I)^{-1}$ and $(\pE T-\zeta
  I)^{-1}$ coincide due to \eqref{opars}. Therefore
$(T_{T}-\zeta I)^{-1}$  is bounded  if and only if  $(\pE T-\zeta I)^{-1}$ is bounded. Thus
$\hat\rho(\pE T)=\hat\rho(T_{T})$ and hence the assertion holds by Proposition \ref{espoper01}.
\item
 For $\zeta \in \rho(T_{T})$, the operator $(T_{T}-\zeta I)^{-1}$
 is bounded in the space $(\mul T)^{\perp}$. By the previous item
$(T-\zeta I)^{-1}$ is also bounded. Thus, taking into account
\eqref{depolr}, one has
\begin{align*}
\dom (T-\zeta I)^{-1}&=\ran (T-\zeta I)\\
&=\oP{\ran (T_{T}-\zeta I)}{\mul T}\\
&=\oP{\dom (T_{T}-\zeta I)^{-1}}{\mul T}\\
&=\oP{(\mul T)^{\perp}}{\mul T}=\H.
\end{align*}
Therefore $\zeta\in \rho(T)$. The other inclusion follows from a similar reasoning.
\item It follows from \eqref{depolr} that
$$\ran (T-\zeta I)=\oP{\ran (T_{T}-\zeta I)}{\mul T}\,.$$
Thus, since $\mul T$ is closed, $\ran (T-\zeta I)$ is closed if and
only if $\ran (T_{T}-\zeta I)$ is closed.  Therefore
$\sigma_c(T)=\sigma_c(T_{T}).$
\item[(d)] Since $\dom T\subset(\mul T)^{\perp}$, one
  has
  \begin{equation*}
    \ker (T-\zeta I)=\ker (T_{T}-\zeta I)\,.
  \end{equation*}
From this equation (d) follows
  at once.
\end{enumerate}
\end{proof}

\section{Dissipative relations}
\label{sec:dissipative-relations}
This section presents the theory of dissipative relations in a fashion
similar to the theory of dissipative operators given in
\cite[Chap.\,4, Sec.\,4]{MR2760647}. The theory of these operators was
introduced in \cite{MR0104919} (see further developments in
\cite{MR0282238}).  This section is related to
\cite[Sec.\,3]{MR0361889} and extends some of its results
(cf. \cite{MR3057107}).

A linear relation $L$ is said to be dissipative if
\begin{equation}
  \label{eq:dissipative-definition}
  \im\ip fg\geq0
\end{equation}
holds for all $\vE fg$ in $L$. If the equality in
\eqref{eq:dissipative-definition} takes place for all $\vE fg$ in $L$,
then $L$ is said to be
symmetric. Note that $L$ is symmetric if and only if $L\subset L^*$.

\begin{theorem}\label{codrih}
  The linear relation $L$ is dissipative if and only if
  the lower half plane $\C_{-}$ is contained in $\hat\rho(L)$ and for all $\zeta\in\C_{-}$,
\begin{equation*}
  \no{(L-\zeta I)^{-1}}\leq -1/{\im \zeta}\,.
\end{equation*}
\end{theorem}
\begin{proof}
Suppose that $L$ is dissipative and let $\zeta \in \C_{-}$. If $\vE
hk\in (L-\zeta I)^{-1}$, i.\,e. $\vE k{h+\zeta k}\in L$, then $\im \ip
k{h+\zeta k}\geq 0$. Therefore
\begin{align*}
0&\leq\im\ip kh+\im \zeta\no k^{2} \\
&\leq|\ip kh|+\im \zeta\no k^{2} \\
&\leq \no h\,\no k+\im \zeta\no k^{2}.
\end{align*}
For $k\neq0$ the last inequality yields \bea\label{didr} \no
k\leq-\frac1{\im \zeta }\no h\,.  \eea If $k=0$, then \eqref{didr}
holds. Hence $\no{(L-\zeta I)^{-1}}\leq -1/{\im \zeta}$ and
$\zeta\in \hat\rho(L)$.

Conversely, if $\vE fg\in L$ and $\tau>0,$ then  $\vE
{g-(-i\tau)f}f\in[L-(-i\tau) I]^{-1}$ and,
by hypothesis,
\begin{align*}
\no f^{2}&\leq\frac1{\tau^{2}}\no{g+\tau if}^{2}\\&\leq\frac1{\tau^{2}}(\no g^{2}+\tau^{2}\no f^{2}+2\tau\im\ip fg)\,.
\end{align*}
Thus, \beao -\frac{1}{2\tau}\no g^{2} \leq \im\ip fg.
\eeao Letting $\tau$ tends to infinity, one arrives at
$\im\ip fg\geq0$ and hence $L$ is dissipative.
\end{proof}

\begin{remark}
  \label{rem:quasiregular-symmetric}
  Note that if $A$ is symmetric so is $-A$. Therefore, by Theorem
  \ref{codrih} $\C_{-}\subset\hat\rho(-A)$, which implies that the
  upper half plane $\C_{+}$ is contained in $\hat\rho(A)$. Hence $A$
  is symmetric if and only if $\C\backslash\R\subset\hat\rho(A)$ and,
  for all $\zeta\in\C\backslash\R$, the inequality
\begin{equation*}
  \no{(A-\zeta I)^{-1}}\leq 1/{|\im \zeta|}
\end{equation*}
holds.
\end{remark}
Due to Theorem \ref{codrih},
the set $\C_{-}$ is a connected component of the quasi-regular set of
any dissipative relation. Hence,
if a dissipative relation is closed, then its
deficiency index (given by \eqref{ranTp}) is constant on
$\C_{-}$. For a closed dissipative relation $L$, define
$\eta_{-}(L):=\eta_{\zeta}(L)$ for any $\zeta\in \C_{-}$. Thus, in view of
\eqref{defspace}, one has
\begin{equation*}
  \eta_{-}(L)=\dim \Nk{\cc{\zeta}}{L^{*}}\,.
\end{equation*}

Furthermore, if $A$ is a closed symmetric relation, then
$\C_{+}$ is also a connected component of $\hat\rho(A)$, and one can
also set $\eta_{+}(A):=\eta_{\overline\zeta}(A)$ for any $\zeta\in \C_{-}$.
Hence $A$ has indices
\begin{equation}
\label{eq:deficiency-indices-symmetric}
(\eta_{+}(A),\eta_{-}(A))=(\dim
\Nk{\zeta}{A^{*}},\dim \Nk{\cc{\zeta}}{A^{*}}),\,\, \zeta\in \C_{-}\,.
\end{equation}

A dissipative relation $L$ is said to be maximal if it is closed and
$\eta_-{(L)}=0$ (or equivalently $\C_{-}\subset\rho(L)$).

$L$ is maximal in the following sense. If $A$ is another dissipative
relation such that $L\subset A$, one verifies that
$\eta_-(\overline{A})\leq \eta_-(L)$. Then $\overline{A}$ is also
maximal. Thus, for $\zeta \in \C_{-}$, one has $(L-\zeta
I)^{-1}\subset(\overline{A}-\zeta I)^{-1}$ and both are in $\cA
B(\H)$. This implies that $(L-\zeta I)^{-1}=(\overline{A}-\zeta
I)^{-1}$ and then $L=\overline{A}.$ Since $A$ is an extension of $L$,
$A=L$. Hence, a maximal dissipative extension does not admit proper
dissipative extensions.

The following assertion is taken from \cite[Lem.\,2.1]{MR3057107} and
sheds light on the interrelationship between $\dom L$ and $\mul L$ for a
dissipative relation $L$.

\begin{proposition}\label{domdisp}
If $L$ is a closed dissipative relation, then $\dom L\subset (\mul L)^{\perp}$.
Moreover, if $L$ is maximal dissipative, then $\cc{\dom L}=(\mul L)^{\perp}$.
\end{proposition}

Due to this proposition, the spectrum of any closed
dissipative relation has the properties given in Theorem
\ref{PopvsRel}.

\begin{proposition}\label{comdew}
  If $L$ is a closed dissipative relation whose domain is the whole
  space, then $L$ is a bounded maximal dissipative operator.
\end{proposition}
\begin{proof}
  If $\vE f{if}\in L^{*}$, then there exists $\vE fg\in L$ such that
  $\ip fg=\ip {if}f$. Thus $-i\no{f}^{2}=\ip fg$ and \beao -\no
  f^{2}=\im (-i\no{f}^{2})=\im \ip fg\geq0\,. \eeao This implies that
  $f=0$ and then
  $\eta_{-}(L)=0$. Furthermore, by Proposition \ref{domdisp},
  \begin{equation*}
    \mul L\subset(\dom L)^{\perp}=\{0\}\,.
  \end{equation*}
  Thereupon $L$ is a closed operator defined on the whole space and therefore it is bounded.
\end{proof}


\begin{proposition}\label{qrpomde}
If $L$ is a maximal dissipative relation, then
\begin{equation*}
  \rho(L)\cap(\C_{-}\cup\R)=\hat\rho(L)\cap(\C_{-}\cup\R)\,.
\end{equation*}
\end{proposition}
\begin{proof}
  Suppose that
  $\zeta \in \hat\rho(L)\cap\C_{-}$. Since $\eta_{-}(L)=0$, the set
  $\ran (L-\zeta I)$ coincides with the whole space. This means that $\zeta \in \rho(L)$.

  Now suppose that $\zeta \in \hat\rho(L)\cap\R$. Since $\hat\rho(L)$
  is open, there exists an open neighborhood $\mathcal{V}(\zeta)$ of $\zeta$ in
  $\hat\rho(L)$. Since $\eta_\zeta(L)$ is constant on each connected
  component of $\hat\rho(L)$, one has, for any
  $\nu\in\mathcal{V}(\zeta)\cap\C_{-}$,
  \begin{equation*}
    \eta_\zeta(L)=\eta_\nu(L)=\eta_{-}(L)=0\,.
  \end{equation*}
  Thus $\ran (L-\zeta I)=\H$, which yields $\zeta\in \rho(L)$.
\end{proof}

From Proposition~\ref{qrpomde}, one concludes that
\bea\label{qrpomdec}
\sigma(L)\cap\R=\hat\sigma(L)\cap\R.
\eea
\begin{proposition}\label{inadco}
A linear relation $L$ is dissipative if and only if  $-L^{-1}$  is dissipative.
\end{proposition}
\begin{proof}
  Suppose that $L$ is dissipative and let $\vE fg\in -L^{-1}$ then
  $\vE{-g}f\in L$ and \beao 0\leq\im\ip{-g}f=\im-\ip{g}f=\im\ip fg,
  \eeao thence $-L^{-1}$ is dissipative. The converse can be
  established by noting that
\begin{equation}
\label{eq:dissipative-aux}
-(-L^{-1})^{-1}=L\,.
\end{equation}
\end{proof}
Thus, the transform $L\to-L^{-1}$ preserves the class of dissipative
relations. Furthermore this transform also preserves the subclass of
maximal, dissipative relations.
\begin{theorem}\label{inadmco}
  If $L$ is a maximal, dissipative relation, then $-L^{-1},$ $-L^*$
  and $-L^{\perp}$ are maximal dissipative relations.
Conversely, if either $-L^{-1},$ $-L^*$ or $-L^{\perp}$ is a maximal
dissipative relation, then $L$ is a maximal dissipative relation.
\end{theorem}
 \begin{proof}
    It follows from Proposition
   \ref{inadco} that
   $-L^{-1} $ is dissipative and since $L$ is closed, so is $-L^{-1}$. One should show that $\Nk
   i{(-L^{-1})^{*}}$ is the trivial relation. Let
   \begin{equation}
     \label{eq:element-in-defect-relation}
     \vE g{ig}\in (-L^{-1})^{*}\,.
   \end{equation}
   The maximality of $L$ means that $\ran (L+iI)=\H$, so there
   exists $\vE h{ig}\in (L+iI)$, which implies that
   $\vE {ig-ih}{-h}\in -L^{-1}$. Taking into account
   \eqref{eq:element-in-defect-relation}, one has
 \begin{align*}
 \ip{g}{-h}&=\ip{ig}{ig-ih}\\&=\no g^{2}+\ip{g}{-h}.
 \end{align*}
Thus $g=0$ and  $-L^{-1}$ is maximal dissipative.

Now consider $\zeta\in \C_{-}$ and let
$\vE hk\in (-L^{*}-\zeta I)^{-1}$, that is
$\vE k{-h}\in (L+\overline \zeta I)^{*}$. Since $\eta_{-}(L)=0$, one
has $\ran
(L+\overline \zeta I)=\H$; so one can find an element
$\vE fk$ in $(L+\overline \zeta I)$. Therefore, it should hold that \bea\label{cosmad} \no
k^{2}=\ip f{-h}.  \eea Observe that
$\vE kf\in [L-(-\overline \zeta) ]^{-1}$ and from Theorem
\ref{codrih}, \bea\label{cosmad1} \no f\leq -\frac1{\im \zeta}\no k.
\eea Then, by \eqref{cosmad} and \eqref{cosmad1},
\begin{align*}
\no k^{2}&=\ip f{-h}\\
&\leq \no f\,\no h\\
&\leq -\frac1{\im \zeta}\no k\, \no h.
\end{align*}
For $k\neq0$ the last inequality yields \bea\label{cosmad2} \no
k\leq-\frac1{\im \zeta }\no h.  \eea If $k=0$, then \eqref{cosmad2} is
trivial. Hence, by Proposition \ref{codrih}, $-L^{*}$ is
dissipative. Maximality, i.\,e. the fact that
$\C_{-}\subset \rho(-L^{*})$, follows from
Proposition~\ref{prop:resolvent-adjoint}.

Observe that
\begin{align*}
-L^\perp&=-(\mF W\mF WL)^{\perp}\\
&=-(\mF WL)^*\\
&=-(-L^{-1})^*\,.
\end{align*}
Thus by what has been proven $-L^\perp$ is maximal dissipative. The
converse assertions follow from \eqref{eq:dissipative-aux},
 $-(-L^{*})^{*}=L$, and $-(-L^{\perp})^{\perp}=L$.
 \end{proof}
Let us turn to the question of whether the sum of maximal
dissipative relations is a maximal dissipative relation.
\begin{theorem}
  \label{thm:sum-of-dissipative-dissipative}
  Let $A$ and $V$ be maximal dissipative relations. If $\dom
  V=\H$, then $L=A +V$ is a maximal dissipative relation.
\end{theorem}
\begin{proof}
  The fact that $L$ is dissipative follows directly from
  \eqref{eq:dissipative-definition}. Closedness is a consequence of
  Remark~\ref{rem:birman-stability-closedness}. It remains to be proven
  that $L$ is maximal, which in turn is reduced to showing that $\Nk i{L^*}$
  is trivial.  Observe that Proposition \ref{comdew} ensures
  $V\in\cA B(\H)$ and therefore $V^{*}\in\cA B(\H)$.  By Proposition
  \ref{linadj}, if $\vE f{if}\in L^{*}$, then there is
  $\vE fs\in A^{*}$ and $\vE ft\in V^{*}$ such that $if=t+s$. Thus
  \bea\label{eq:ifplusts}\vE ft\in V^{*}\,,\quad \vE f{if -t}\in A^{*}\,.\eea On the other
  hand, since $-i\in \rho(A)$, there exists $\vE tk\in(A+i I)^{-1}$,
  which implies that $\vE k{t-ik}\in A$. This inclusion and
  the second one in \eqref{eq:ifplusts} yield $\ip{if-t}k=\ip f{t-ik}$ and therefore
  $\im\ip kt=\im\ip ft$. Thus, one obtains from the dissipativity
  condition that
  \begin{equation}
    \label{eqref:posiimft}
    0\leq\im\ip k{t-ik}\le\im\ip kt=\im\ip ft\,.
  \end{equation}
   By Theorem
  \ref{inadmco}, $-V^{*}$ is dissipative. Using this fact
  and the
  first inclusion in \eqref{eq:ifplusts} one arrives at \beao \im\ip
  ft=-\im\ip f{-t}\leq 0\,, \eeao which, together with
  \eqref{eqref:posiimft}, yields $\im\ip ft=0$. To conclude the proof,
  use the dissipativity of  $-A^{*}$ (Theorem~\ref{inadmco}) and
  the second inclusion in \eqref{eq:ifplusts} to obtain
  \beao 0\leq \im\ip
  f{-if+t}=-\no f^{2}\,, \eeao which implies $f=0$.
 \end{proof} 

 Let us introduce the concept of relative boundedness for relations
 in a way analogous to the same concept for operators \cite[Chap. 4,
 Sec. 1]{MR0407617}.

 A relation $S$ is said to be $T$-bounded if $\dom T\subset \dom S$
 and there exists $c>0$ such that for all $\vE fh\in T$ and
 $\vE fg\in S$ the following holds \bea\label{boundC} \no g\leq
 c\noa{\vE fh}\,.  \eea Observe that if $S$ is $T$-bounded, then $S$
 is an operator.  Furthermore, $S$ is said to be strongly $T$-bounded
 when $c<1$ in \eqref{boundC}. Note that our definition of strong
 relative boundedness is formally stronger than the definition given
 in \cite[Chap. 4, Sec. 1]{MR0407617}, however it can be proven to be
 equivalent by following the argumentation of the proof of
 \cite[Thm.\,3 Sec.\,4 Chap\,3]{MR1192782}.
\begin{lemma}\label{cpreisum}
Let $S$ be strongly $T$-bounded. The relation $T$ is closed  if and only if $T+S$ is closed.
\end{lemma}
\begin{proof}
  Since $S$ is strongly $T$-bounded, it follows from the triangle
  inequality that there exists $0<c<1$ such that for all $\vE fh\in T$
  and $\vE fg\in S,$ \bea\label{kcttofnoe0} (1-c)\noa{\vE fh}\leq\noa
  {\vE f{h+g}}\leq(1+c)\noa {\vE f{h}}.  \eea

  If $T$ is closed and $\vE fs \in \overline{T+S}$, then
  there are sequences $\llb\vE {f_{n}}{h_{n}}\rrb_{n\in \N}$ in $T$
  and $\llb\vE {f_{n}}{g_{n}}\rrb_{n\in \N}$ in $S$ such that \beao
  \vE {f_{n}}{h_{n}+g_{n}}\rightarrow\vE fs\,.  \eeao

  It follows from \eqref{kcttofnoe0} and the fact that
  $\llb \vE {f_{n}}{h_{n}+g_{n}}\rrb_{n\in \N}$ is a Cauchy sequence,
 that
  $\llb \vE {f_{n}}{h_{n}}\rrb_{n\in\N}$ converges to some
  $\vE fh\in T$. Thereupon, there exists $\vE fg\in S$ such that
  $\vE f{h+g}\in T+S$. Thus, again by \eqref{kcttofnoe0}, one obtains
\begin{align*}
\noa{\vE f{h+g}-\vE fs}&=\lim_{n\rightarrow \infty}\noa{\vE f{h+g}-\vE {f_{n}}{h_{n}+g_{n}}}\\
&=\lim_{n\rightarrow \infty}\noa{\vE {f-f_{n}}{(h-h_{n})+(g-g_{n})}}\\
&\leq\lim_{n\rightarrow \infty}(1+c)\noa{\vE{f-f_{n}}{h-h_{n}}}=0\,.
\end{align*}
Hence $\vE fs\in T+S$, which establishes that $T+S$ is closed. The
proof of the converse assertion is carried out analogously.
\end{proof}
The requirement of $S$ being strongly $T$-bounded in the last result
cannot be relaxed (see a counterexample in \cite[Sec.\,4 Chap\,3]{MR1192782}).
\begin{lemma}\label{dofadjisb}
Let $T$ be a closed linear relation. If $S$ and $S^*$ are strongly
$T$-bounded and strongly $T^*$-bounded, respectively, then
\bea\label{doaditb}
(T+S)^{*}=T^{*}+S^*\,.
\eea
\end{lemma}
\begin{proof}
  Due to \eqref{eq:inclusion-sum-adjoints-adjoint-sum},
  $T^{*} + S^*\subset \mF W( T+ S)^\perp$. It follows from Lemma
  \ref{cpreisum} that $\mF W (T+S)$ and $(T^{*}+S^*)$ are closed. Thus, for
  proving \eqref{doaditb}, it suffices to show that
  \bea\label{eqalformad} \oP{\mF W(T+ S)}{(T^{*}+ S^*)}=\HH\,.  \eea
  By hypothesis, there exist $0<b<1$ such that, for any
  $\vE f{h}\in T$, $\vE fg\in S$, $\vE lt\in T^{*}$ and
  $\vE l{s}\in S^{*}$, the inequalities
\begin{equation}\label{01eqtp}
\no {g}^{2}\leq b\noa{\vE fh}^{2}\quad\text{ and }\quad
\no{s}^{2}\leq b\noa{\vE lt}^{2}
\end{equation}
hold. On the other hand one obtains from \eqref{epoar}, using the fact
that $T$ is closed, that \bea\label{deotTa} \oP{\mF W T}{T^{*}}=\HH\,.
\eea Thus, for every $\vE rk\in \HH,$ there exist $\vE fh\in T$ and
$\vE lt\in T^{*}$ such that $\vE rk=\vE{-h+l}{f+t}$. Since $\dom
T\subset \dom S$ and $\dom T^*\subset \dom S^*$, one can find
$g,s\in\H$ such that $\vE fg\in S$ and $\vE ls\in S^{*}.$ Define the
linear relation $\mF Q$ in $\oP{(\HH)}{(\HH)}$ as follows \beao \mF
Q:=\llb\vE {\tilde r}{\tilde s}\ :\ \tilde r=\vE rk \mbox{ and }
\tilde s=\vE {-g}s \rrb\,.\eeao Due to the fact that the norm in $\HH$
is equivalent to \eqref{noaHH}, it follows from \eqref{01eqtp} and
\eqref{deotTa} that, for any $\vE {\tilde r}{\tilde s}\in \mF Q$,
\begin{align*}
\no {\tilde s}^{2}&=\no {g}^{2}+\no s^{2}\\
&\leq b\left(\noa{\vE{f}{h}}^{2}+\noa{\vE lt}^{2}\right)\\
&\leq b\left(\noa{\vE{-h}{f}}^{2}+\noa{\vE lt}^{2}\right)\\&=b\noa{\vE{-h+l}{f+t}}^{2}=b\no{\tilde r}^{2}\,.
\end{align*}
Then $\mF Q\in \cA B(\HH)$ with $\no{\mF Q}<1,$ which implies that
\beao
\ran(\mF Q+\mF I)=\HH.\eeao
Therefore, for any $\vE vw\in \HH,$ there exists
$\vE {\tilde r}{\tilde s}\in \mF Q$ such that
\begin{align*}
\vE vw&=\tilde s+\tilde r\\&=\vE {-g-h+l}{s+f+t}\\&=\mF W\vE f{h+g}+\vE l{t+s}\in\oP{\mF W(T+S)}{(T^{*}+S^*)}\, ,
\end{align*}
whence \eqref{eqalformad} follows.
\end{proof}

In order to state the following assertion, let us introduce a subclass
of the class of symmetric
relations. A relation $A$ is said to be positive (denoted by $A\geq0$) whenever
\beao
\ip fg\geq0\quad \text{ for all }\quad \vE fg\in A\,.
\eeao

\begin{theorem}\label{psewpos}
Let $A$ and $B$ be two selfadjoint relations such that $B$ is positive and strongly $A$-bounded. Then $A+iB$ is a maximal dissipative relation.
\end{theorem}
\begin{proof}
  By a direct verification of \eqref{eq:dissipative-definition}, one
  establishes that $A+iB$ is dissipative. The closedness
  follows from Lemma \ref{cpreisum} after noting that $iB$ is
  also strongly $A$-bounded.

  It remains to prove that $\Nk i{(A+iB)^{*}}$ is a trivial
  relation. By Lemma \ref{dofadjisb}, one has
  $(A+iB)^{*}=A-iB$. For an arbitrary \beao \vE f{if}\in A-iB\,, \eeao there exist
  $\vE fh\in A$ and $\vE fg\in B$ such that $if=h-ig.$ Thus, $\vE
  f{i(f+g)}\in A$ and, due to the selfadjointness of $A$, one concludes
  \beao \no f^{2}+\ip fg=\im \ip f{i(f+g)}=0\,.  \eeao Since $B$ is
  positive, the last equality yields that $f=0.$ Therefore $A+iB$ is maximal dissipative.
 \end{proof}




\begin{remark}
  \label{rem:dissipative-properties-with-mul}
One can verify that the operator part of a closed
dissipative relation is a closed dissipative operator. Conversely, for
a closed relation $L$ such that $\dom L \subset (\mul L)^{\perp}$, if
$\pE L$ is dissipative, then $L$ is dissipative.
\end{remark}

\begin{proposition}\label{lsimresin}
  Let $L$ be a closed linear relation.  If $L$ is (maximal)
  dissipative, then $L_{L}$ is (maximal) dissipative operator in
  $\mm{\mul L}$ and \bea\label{indeq} \eta_-(L_{L})=\eta_-(L).  \eea
  Conversely, if $\mul L\subset (\dom L)^{\perp}$ and $L_{L}$ is
  (maximal) dissipative, then $L$ is (maximal) dissipative and, therefore,
  \eqref{indeq} holds.
\end{proposition}
\begin{proof}
  Suppose that $L$ is closed dissipative. It follows from Proposition
  \ref{domdisp} and \eqref{opars} that $L_{L}$ is a closed,
  dissipative operator in $\oP{(\mul L)^{\perp}}{(\mul L)^{\perp}}$.
  Moreover, \eqref{depolr} implies that
 \begin{align*}\label{Lndirest}
  \oM{\H}{\ran(L-\zeta I)} &=\oM{\H}{[\oP{\ran(L_{L}-\zeta I_{L})}{\mul L}]}\\
  &=\oM{[\oM{\H}{\mul L}]}{\ran(L_{L}-\zeta I_{L})}\\&=\oM{(\mul L)^{\perp}}{\ran(L_{L}-\zeta I_{L})}\,.
 \end{align*}
Whence \eqref{indeq} follows.

For the converse assertion, one again uses \eqref{opars} to conclude
that $\pE L$ is dissipative. Thus, taking into account
Remark~\ref{rem:dissipative-properties-with-mul}, one has that $L$ is
dissipative. Finally, due to \eqref{indeq}, $L$ is maximal if and only
if $L_L$ is maximal.
\end{proof}

\section{Dissipative extensions of dissipative relations}
\label{sec:dissipative-extensions}

This section is devoted to the development of the theory of extensions
of dissipative relations. We consider only extensions without exit to
a larger space (cf. \cite[Appendix\,1]{MR1255973}). Our approach is
similar to the one used in the von Neumann theory. There are other
ways of dealing with extensions of operators (see for instance
\cite[Sec.\,14]{MR2953553}).

A relation $V$ is a contraction if it is bounded (and then it is
actually an operator)
with $\no V\leq 1$. It is known that if a relation $V$ satisfies $V^{-1}\subset
V^{*}$, then $V$ is a particular kind of contraction called isometric
operator for which $\no V=1$ holds. Moreover if $V^{-1}=V^{*}$ the
operator $V$ is said to be unitary.

Denote $\tE_e:=\{\zeta\in\C\ :\ |\zeta|>1\}$. For any contraction $V$,
one verifies that $\tE_{e}\subset\hat\rho(V)$. Therefore, if $V$ is a
closed contraction, then the deficiency index
$\eta_{\zeta}(V)$ (see \eqref{ranTp}) is constant when $\zeta$ runs
through $\tE_e$. In this case, define
\begin{equation*}
  \eta_e(V):=\eta_{\zeta}(V), \ \ \zeta\in \tE_{e}\,.
\end{equation*}
Following the argumentation of \cite[Thm.\,4.2.2]{MR1192782} (which
deals with isometric operators), it can
be proven that, for any closed contraction $V$,
\bea\label{diocc}
\eta_e(V)=\dim{(\oM{\H}{\dom {V}})}.
\eea
If $\hat V$ is also a closed contraction such that $V\subset\hat V$, then
\bea\label{idocceos}
\eta_{e}(V)=\eta_{e}(\hat V)+\eta_{0},
\eea
where $\eta_{0}=\dim (\oM{\dom \hat V}{\dom V})$.

A contraction $V$ is said to be maximal if it is closed and
$\eta_e{(V)}=0$. A maximal contraction does not admit
contractive extensions, that is, extensions that are also
contractions. This will be clear from the following statement, which
is related to \cite[Sec.\,4.4]{MR1192782} and
\cite[Thm.\,5.1]{MR0361889}.
\begin{theorem}\label{extisocon}
  Let $V$ be a closed contraction. The operator $\hat V$ is a closed
  contractive extension of $V$ if and only if there exists a unique
  closed contraction $W$ such that
 \begin{equation}\label{etocex}
 \hat V=\oP VW,
\end{equation}
and
 \begin{equation}\label{etocex00}
 2|\re(\ip fh-\ip gk)|\leq (\no f^{2}-\no g^{2}) +(\no h^{2}-\no k^{2}),
\end{equation}
for all $\vE fg\in V$ and $\vE hk\in W$. Note that the right-hand side
of \eqref{etocex} is an orthogonal sum of relations (see
(\ref{eq:sets-relations-operations})). Moreover, if $V$ is isometric,
then the condition \eqref{etocex00} turns into the condition that
either
\begin{equation}\label{eq:dom-range-perp}
  \dom V\perp\dom W\quad\text{ or }\quad\ran V \perp\ran W
\end{equation}
holds. In view of \eqref{etocex}, the conditions in
(\ref{eq:dom-range-perp}) hold simultaneously.
 \end{theorem}
\begin{proof}
  Suppose that $\hat V$ is a closed contractive extension of $V$ and
  consider $W=\oM{\hat V}V$ then $W$ is a closed contraction and one
  verifies that $\hat V=\oP VW$.

  For every $\vE fg\in V$ and $\vE hk\in W$, one has that
  $\vE {\alpha f+h}{\alpha g+k}\in \hat V$ and
  $\no {\alpha g+k}\leq\no {\alpha f+h}$ with $\alpha\in \C.$ Then
\beao
|\alpha|^{2}\no g^{2}+\no k^{2}+2\re\overline\alpha\ip gk&=&\no{\alpha g+k}^{2}\\
&\leq&  \no{\alpha f+h}^{2}\\
&=&|\alpha|^{2}\no f^{2}+\no h^{2}+2\re\overline\alpha\ip fh,
\eeao
whence
\bea\label{etocex01}
-2\re\overline\alpha(\ip fh-\ip gk)\leq |\alpha|^{2}(\no f^{2}-\no g^{2}) +(\no h^{2}-\no k^{2}).
\eea
Thus, setting $\alpha:=\pm 1$, the inequality \eqref{etocex00} holds.

If $V$ is isometric then $\no f=\no g$. It turns out that in this case
\begin{equation}\label{eq:aux-cond-contraction}
\ip fh=\ip gk
\end{equation}
since
otherwise there would exist $\tau>0$ such that $$\tau |\ip fh-\ip
gk|>\no h^{2}-\no k^{2}.$$
This inequality contradicts \eqref{etocex01} when $\alpha=-\tau|\ip
fh-\ip gk|/( \ip hf-\ip kg)$. Therefore, since $V$ and $W$ are
orthogonal, it follows from \eqref{eq:aux-cond-contraction} that
\beao
0&=&\ipa{\vE fg}{\vE hk}\\&=&\ip fh+\ip gk=2\ip fh=2\ip gk.
\eeao
The uniqueness of the decomposition is trivial. The converse assertion
is straightforward.
\end{proof}

Note that under the assumption that $V$ is isometric in
\eqref{etocex}, the number $\eta_{0}$ in (\ref{idocceos}) is given by
$\eta_{0}=\dim\dom W$. Moreover, in this case, $\hat V$ is isometric
if and only if $W$ is isometric.

We now turn to the question of extending closed dissipative relations and, in
particular, closed symmetric relations. To this end, we introduce a
fractional linear transformation of a relation as follows.
\begin{definition}
  Following \cite{MR0361889}, for a relation $T$ and $\zeta \in \C$,
  define the Cayley transform of $T$ by \beao \CA \zeta T:=\llb \vE
  {g-\overline{\zeta}f}{g-\zeta f}\ :\ \vE fg \in T\rrb
  =I+(\cc{\zeta}-\zeta)(T-\cc{\zeta}I)^{-1}\eeao

Also, let us define the $Z$ transform of $T$ (cf. \cite{MR2093073})
\beao
\CZ \zeta T:= \cc{\zeta}\CA \zeta T
\eeao
This is a linear relation which satisfies
\begin{align}\label{CtoT}
 \dom \CZ \zeta T=\ran (T-\overline{\zeta}I)\,,& \quad
 \ran \CZ \zeta T=\ran (T-\zeta I)\,,\\ \mul \CZ \zeta T=\ker
 (T-\overline{\zeta} I)\,,& \quad \ker \CZ \zeta T=\ker (T-\zeta I)\,.
\end{align}
\end{definition}

The $Z$ transform has the following properties (see \cite[Lems.\,2.6,
2.7]{MR0361889} and \cite[Props.\,3.6, 3.7]{MR2093073}).
For any $\zeta \in\C$:
\begin{enumerate}[{(i)}]
 \item \label{cay01} $\CZ \zeta {\CZ \zeta T}=T$.
 \item $\CZ \zeta T\subset \CZ \zeta S\ \Leftrightarrow\ T\subset S$.\label{tres}
 \item $\CZ {-\zeta} T=-\CZ \zeta {-T}$.
 \item If $\abs{z}=1$, then $\CZ \zeta {T^{-1}}=\CZ {\overline {\zeta}} T=(\CZ \zeta {T})^{-1}$.
  \end{enumerate}
For any $\zeta \in \C\backslash \R$:
  \begin{enumerate}[{(i)}]
  \setcounter{enumi}{4}
 \item $\CZ { \zeta} {T\dotplus S}=\CZ { \zeta} T\dotplus \CZ { \zeta} S$.
 \item \label{cay07} If  $\zeta=\pm i$, then $\CZ\zeta{\oP TS}=\oP{\CZ \zeta T}{\CZ \zeta S}$.
 \item\label{cay05} $\CZ { \zeta} {T^*}=(\CZ {\overline{\zeta}} {T})^{*}$.
 \item\label{cay10} $\CZ \zeta T$ is closed $\Leftrightarrow$ $T$ is closed.
\end{enumerate}
\begin{proposition}\label{caydac}
  Under the assumption that $\zeta \in\C_{+}$ and $\abs{\zeta}=1$, the linear relation $L$
  is (closed, maximal) dissipative if and only if $V=\CZ\zeta L$ is a
  (closed, maximal) contraction.
\end{proposition}
\begin{proof}
Suppose that $L$ is a dissipative relation and let
$\vE{g-\overline\zeta f}{\cc{\zeta}g-f}\in\CZ\zeta L=V$, with $\vE fg\in L$. Then
\bea\label{dvsis}
\no {g-\overline\zeta f}^{2}-\no {\cc{\zeta}g-f}^{2}&=&2\re(-\zeta\ip fg)+
2(\re\overline\zeta\ip fg),\nonumber \\
&=&4(\im \zeta)\im \ip fg\geq 0.
\eea
Thus $V$ is a contraction.

Conversely, let $\vE{g-\cc{\zeta}f}{\overline\zeta g-f}\in\CZ\zeta V=L,$ with $\vE fg\in V.$ Then
\bea\label{dvsis0}
\im\ip{g-\cc{\zeta}f}{\overline\zeta g-f}&=&\im(\overline\zeta\no g^{2}
+\zeta\no f^{2}-2\re\ip gf)\nonumber \\
&=&\im \zeta (\no f^{2}-\no g^{2})\geq 0,
\eea
therefore $L$ is dissipative. Note that $L$ and $V$ are simultaneously
closed due to (\ref{cay10}). As regards the maximality,
\begin{align*}
\eta_{e}(V)&=\dim(\oM{\H}{\dom V})\\
&=\dim(\oM{\H}{\dom\CZ \zeta L})\quad(\text{due to \eqref{CtoT}}) \\
           &=\dim(\oM{\H}{\ran(L-\overline\zeta I)})=\eta_{-}(L)\,.
\end{align*}

\end{proof}
\begin{remark}\label{rctosti}
 It follows from \eqref{dvsis} and \eqref{dvsis0} that, for all
 $\abs{\zeta}=1$, a relation is symmetric if and only if its $Z$ transform is isometric.
 Moreover,  Proposition~\ref{caydac} shows that the $Z$
 transform gives a one-to-one correspondence between contractions and
 dissipative relations.
\end{remark}

\begin{proposition}
  Let $L$ be a closed dissipative relation. Then $\hat L$ is a closed
  dissipative extension of $L$ if and only if there exists a unique
  closed dissipative relation $S$ such that \beao \hat L=\oP{L}{S},
  \eeao and, for all $\vE fg\in L$ and $\vE hk\in S$,
  $$\im (\ip fg+\ip hk)\geq|\im(\ip fk-\ip gh)|.$$
\end{proposition}
\begin{proof}
To prove the assertion, one applies the $Z$ transform at $i$ to (\ref{etocex}) and
(\ref{etocex00}) taking into account property (vi).
\end{proof}

Let us now turn to the question of generalizing the so-called second von
Neumann formula (cf. \cite[Thm.\,6.2]{MR0361889}). This generalization
cannot be achieved for dissipative extensions of an arbitrary dissipative
relation since, in general, the conditions (\ref{eq:dom-range-perp})
are not satisfied.

The next assertion can be found in \cite[Thm.\,6.1]{MR0361889} and
 corresponds to the first von Neumann formula. It characterizes the
 adjoint of a closed symmetric relation by means of its deficiency
 space \eqref{defspace}.  We omit the proof since it can be obtained
 by the same argumentation used in the proof of the first von Neumann
 formula (cf. \cite[Thm.\,4.4.1]{MR1192782}).
\begin{theorem}\label{Fovon00}
 For a closed symmetric relation $A$, one has
 \bea\label{von01}
 A^*=
 A\dotplus \Nk{\overline{\zeta}}{A^*}\dotplus \Nk \zeta {A^*}\,,\qquad\zeta\in\C\backslash\R.
 \eea
 For $\zeta\in \{i,\,-i\}$, the direct sum in \eqref{von01} is orthogonal.
\end{theorem}

The following assertion is a generalization of the second von Neumann formula.
\begin{theorem}\label{Fovon01d}
  Let $A$ be a closed symmetric relation. $\hat A$ is a closed
  dissipative (symmetric) extension of $A$ if and only if, for a fixed
  $\zeta\in\C_{+}$ ($\zeta\in \C\backslash\R$), \bea\label{von02d}
  \hat A=A\dotplus (\mF V-\mF I)D, \eea where
  $D\subset \Nk \zeta {A^*}$ is a closed bounded relation and
  $\mF V:D\rightarrow \Nk {\overline{\zeta}} {A^*}$ is a closed
  contraction (isometry) in $\oP{(\HH)}{(\HH)}$.  For $\zeta=i$, the
  direct sum in \eqref{von02d} is orthogonal.
 \end{theorem}
\begin{proof}
  It follows from
  Proposition~\ref{caydac} that $\CZ{\zeta/\abs{\zeta}}{\abs{\zeta}^{-1}A}$ and $\CZ{\zeta/\abs{\zeta}}{\abs{\zeta}^{-1}\hat{A}}$
  are, respectively,  a closed isometric and a closed contraction (isometry) whenever
  $\zeta \in \C_{+}$ ($\zeta\in \C\backslash\R$). Moreover, since
  $A\subset \hat A$, one has  $\CZ{\zeta/\abs{\zeta}}{\abs{\zeta}^{-1}A}\subset\CZ{\zeta/\abs{\zeta}}{\abs{\zeta}^{-1}\hat{A}}$
  in view of property (\ref{tres}).

Theorem \ref{extisocon} implies the existence of a closed contraction (isometry) $W$ such that
 \bea\label{von001co}
  \CZ{\zeta/\abs{\zeta}}{\abs{\zeta}^{-1}\hat{A}}=\oP{\CZ{\zeta/\abs{\zeta}}{\abs{\zeta}^{-1}A}}{W},
 \eea
 where, due to (\ref{eq:dom-range-perp}), 
\bea\label{doanroW}
\begin{split}
\dom W&\subset \oM{\H}{\dom \CZ{\zeta/\abs{\zeta}}{\abs{\zeta}^{-1}A}}=\oM{\H}{\ran(A-\overline{\zeta}I)}=\ker (A^*-\zeta I),\\
 \ran W&\subset \oM{\H}{\ran \CZ{\zeta/\abs{\zeta}}{\abs{\zeta}^{-1}A}}=\oM{\H}{\ran(A-\zeta I)}=\ker (A^*-\overline{\zeta} I).
\end{split}
\eea By applying the $Z$ transform to \eqref{von001co}, using (\ref{cay01}), one
obtains \bea\label{von002co} \hat A=A\dotplus\abs{\zeta}\CZ {\zeta/\abs{\zeta}} W\,.  \eea
Observe that $\dom W$ is closed. Consider the linear relation
\bea\label{daux00} D=\llb \vE v{\zeta v}\ :\ v\in\dom W\rrb, \eea
whence, in view of \eqref{doanroW}, $D\subset \Nk \zeta {A^*}$. Thus $D$
is bounded and then closed.

For every $\vE vw\in W$, define the relation $\mF V$ in
$\oP{(\HH)}{(\HH)}$ with $\dom\mF V=D$ such that
$$\mF V \vE v{\zeta v}=\frac{\zeta}{\abs{\zeta}}\vE w{\overline{\zeta} w}\,.$$
It follows from  \eqref{doanroW} that
$\mF VD\subset\Nk {\overline\zeta} {A^*}$ and, since $W$ is a contraction (isometry):
\bea\label{eoiss00}
\noa{\frac{\zeta}{\abs{\zeta}}\vE w{\overline{\zeta} w}}&=&\no w +\no{\overline{\zeta} w}\nonumber\\
 &\leq&\no v+\no{\zeta v}=\noa {\vE v{\zeta v}}.
\eea
Thus $\mF V$ is a closed contraction (isometry because the equality
holds in \eqref{eoiss00} when $W$ is an isometry). Hence
\begin{equation}
  \label{eq:z-transform-of-w}
  \begin{split}
 \abs{\zeta}\CZ {\zeta/\abs{\zeta}} W&=\llb \vE{\frac{\zeta}{\abs{\zeta}}w-v}{\abs{\zeta}w-\zeta v}\ :\ \frac{\zeta}{\abs{\zeta}}\vE vw\in W \rrb\\
 &=\llb \mF V \vE{v}{\zeta v}  - \vE{v}{\zeta v} \ :\ \vE{v}{\zeta v}\in D\rrb=(\mF V-\mF I)D\,.
\end{split}
\end{equation}
Therefore \eqref{von002co} is transformed into
 $\hat A=A\dotplus (\mF V-\mF I)D$. For $\zeta=i$, the orthogonality of
 the direct sum in \eqref{von002co} follows from property (\ref{cay07}).

We now prove the converse assertion. Define
$$W=\llb \frac{\zeta}{\abs{\zeta}}\vE vw\ :\ \vE v{\zeta v}\in D \mbox{ and } \frac{\zeta}{\abs{\zeta}}\vE w{\cc{\zeta} w}\in \mF VD\rrb\,.$$
Since $\mF V$ is a contraction (isometry), one has
\bea\label{coWwlc}
\noa {\frac{\zeta}{\abs{\zeta}}v}-\noa {\frac{\zeta}{\abs{\zeta}}
  w}&=&\frac1{1+|\zeta|}\left[(1+|\zeta|)\no v-(1+|\zeta|)\no w\right]\nonumber \\
&=&\frac1{1+|\zeta|}\left(\noa {\vE v{\zeta v}}-\noa{\frac{\zeta}{\abs{\zeta}}\vE w{\overline{\zeta} w}}\right)\geq0.
\eea
From this, taking into account that $\dom W=\dom D$, one concludes
that $W$ is a closed contraction (isometry because the equality holds
in \eqref{coWwlc} when $V$ is an isometry).

Also, reading \eqref{eq:z-transform-of-w} backwards, one arrives at
\eqref{von002co}. Now, multiply \eqref{von002co} by $\abs{\zeta}^{-1}$
and apply $\CZ{\zeta/\abs{\zeta}}{\cdot}$ to both sides of the resulting
equality. This yields \eqref{von001co}, where the orthogonality
is a consequence of
\begin{align*}
 \dom W&\subset\ker (A^*-\zeta I)=\oM{\H}{\ran(A-\overline{\zeta}I)}= \oM{\H}{\dom \CZ{\zeta/\abs{\zeta}} {\abs{\zeta}^{-1}A}},\\
 \ran W&\subset\ker (A^*-\overline\zeta I)=\oM{\H}{\ran(A-{\zeta}I)}= \oM{\H}{\ran \CZ{\zeta/\abs{\zeta}} {\abs{\zeta}^{-1}A}}\,.
\end{align*}
The assertion then follows from \eqref{von001co} in view of
Theorem~\ref{extisocon} and Proposition \ref{caydac}.
\end{proof}

As a
consequence of \eqref{von02d}, any dissipative extension $S$ of
a symmetric relation $A$ satisfies $A\subset S\subset A^*$ .

\begin{corollary}\label{extresind0}
  If $A$ is a closed symmetric relation and $\hat A$ is a closed
  dissipative extension of $A$, then \bea\label{fodiidr}
  \eta_{-}(A)=\eta_{-}(\hat A)+\dim[\hat A/A].  \eea
\end{corollary}
\begin{proof}
  In the proof of Theorem \ref{Fovon01d} one verifies that
  $(\mF V-\mF I)$ gives a one-to-one correspondence. Thus, by
  \eqref{von02d} and by \eqref{daux00},
\begin{align*}
 \dim[\hat A/A]&=\dim [(\mF V-\mF I)D]\\
 &=\dim(\dom W)\,.
\end{align*}
Hence, taking into account \eqref{CtoT}, it follows from
\eqref{von001co} and \eqref{idocceos} that
\begin{align*}
\eta_-(A)&=\eta_e(\CZ i A)\\
&=\eta_e(\CZ i {\hat A})+\dim(\dom W)\\
&=\eta_-(\hat A)+\dim[\hat A/A].
\end{align*}
\end{proof}
Since $\Nk{\zeta}{A^{*}}=\Nk{-\zeta}{-A^{*}}$, for a closed symmetric
relation $A$, the equality $\eta_{+}(A)=\eta_{-}(-A)$ holds. This,
together with Corollary \ref{extresind0}, yields that if $\hat A$ is a
closed symmetric extension of $A$, then \bea\label{fodiidrsy}
\eta_{\pm}(A)=\eta_{\pm}(\hat A)+\dim[\hat A/A].  \eea A closed
symmetric relation $A$ is selfadjoint if and only if it has indices
$(0,0)$ (see \eqref{eq:deficiency-indices-symmetric}). For this reason, the selfadjoint relations are maximal
dissipative.

There is another way to construct maximal dissipative extensions of
symmetric relations on the basis of formula \eqref{von01}.
\begin{proposition}\label{mdrofsrn}
  Let $A$ be a closed symmetric relation with finite deficiency index
  $\eta_{-}(A)=n$. Then, for every $\zeta\in\C_{+}$ fixed, the relation
  \bea\label{densa0} \hat A:=A\dotplus\Nk\zeta{A^{*}} \eea is the unique
  maximal dissipative extension of $A$ such that
  $\dim\Nk\zeta{\hat A}=n$.
\end{proposition}
\begin{proof}
  Fix $\zeta \in \C_{+}$. Note that \eqref{von01} implies
  $$A\cap\Nk\zeta{A^{*}}=\left\{\vE 00\right\}\,.$$
Also, \eqref{von01} and \eqref{densa0} yield $\Nk
\zeta {A^*}=\Nk\zeta{\hat A}$.

  Appealing to \eqref{eq:dissipative-definition}, one verifies that $\hat A$ is dissipative. Since $A$ is closed
  and $\dim \Nk{\zeta}{A^{*}}=\eta_{-}(A)$ is finite, $\hat A$ is
  closed. Besides, from Corollary \ref{extresind0}, one has
\begin{align*}
\eta_{-}(\hat A)&=\eta_{-} (A)-\dim [\hat A/A]\\&=\eta_{-} (A)-\dim \Nk\zeta{A^{*}}=0\,.
\end{align*}
Thus, $\hat A$ is a maximal dissipative extension of $A$.

To prove uniqueness, let $L$ be a maximal dissipative extension of $A$
such that $\Nk\zeta{L}=n$. Since $L\subset A^*$,  $\Nk
\zeta {L}\subset\Nk \zeta {A^*}$ holds, so taking into account the
dimension of the spaces, one concludes that $\Nk \zeta {L}=\Nk \zeta
{A^*}$.  Therefore
$\hat A=A\dotplus\Nk\zeta{A^{*}}\subset L$. To complete the proof, it
only remains to recall that $\hat A$ is maximal.
\end{proof}
Note that $\hat A$ is a maximal dissipative, nonselfadjoint
relation. The next assertion complements the previous one.
\begin{proposition}\label{mdrofsrn0}
  Let $A$ be a closed symmetric relation with finite deficiency
  indices $(n,n)$. If $\alpha\in\hat\rho(A)\cap\R$, then
  \bea\label{saeosrws} L:=A\dotplus\Nk\alpha{A^{*}} \eea is the unique
  maximal dissipative extension of $A$ such that
  $\dim\Nk{\alpha}{L}=n$. Moreover $L$ is selfadjoint.
\end{proposition}
\begin{proof}
  Since $A$ is symmetric, it follows that
  $(\C\backslash\R)\cup\{\alpha\}$ is in a connected component of
  $\hat\rho(A)$. Thus $\dim \Nk\alpha{A^{*}}=n$.

If one assumes that $\vE f{\alpha f}\in A$, then
$\vE 0f\in (A-\alpha I)^{-1}$. It follows from the fact that
$(A-\alpha I)^{-1}$ is an operator that $f=0$. Hence $A$ and $\Nk\alpha{A^{*}}$ are linearly
independent.

Taking into account that $\alpha\in \R$, one verifies that $L$ is
symmetric and closed directly from \eqref{saeosrws}. Hence,
$L\subset A^*$. Using again \eqref{saeosrws}, one concludes that
$\Nk\alpha{L}=\Nk\alpha{A^{*}}$.

As in the proof of
Proposition~\ref{mdrofsrn}, one obtains, on the basis of
\eqref{fodiidrsy}, that $\eta_{\pm}(L)=0$. Uniqueness can also be
proved along the lines of the proof of Proposition~\ref{mdrofsrn}.
\end{proof}
Similar to the operator case, one can characterize the spectrum of a
selfadjoint extension of the symmetric relation $A$ in the intervals
intersecting $\hat\rho(A)$.
\begin{proposition}\label{lainesp}
  Let $A$ be a closed symmetric relation with finite deficiency
  indices $(n,n)$ (see \eqref{eq:deficiency-indices-symmetric}) and
  $L$ be a selfadjoint extension of $A$. If a real interval
  $\Delta$ is in $\hat\rho(A)$, then the spectrum of $L$ in $\Delta$
  only has isolated eigenvalues of multiplicity at most $n$.
\end{proposition}
\begin{proof}
  Fix $\zeta\in\sigma(L)\cap\Delta$. Since $\zeta\in\hat\rho(A)$,
  $\dim \Nk\zeta{A^{*}}=n$ and, in view of
  Proposition~\ref{acot-cerr1}, $\ran (A-\zeta I)$ is closed. Now, due
  to the fact that $L\subset A^*$, one can define \bea\label{safese}
  K:=\oM{\ran(L-\zeta I)}{\ran (A-\zeta I)}\,.  \eea Thus
  $K\subset\ker(A^{*}-\zeta I)$ and \bea\label{safese00}
\begin{split}
\dim K&\leq\dim [\ker(A^{*}-\zeta I)]\\&=\dim \Nk\zeta{A^{*}}=n.
\end{split}
\eea
Then by \eqref{safese} and \eqref{safese00}
$\ran(L-\zeta I)$ is closed. This implies that
$\zeta\notin\sigma_{c}(L)$. Furthermore, by
Proposition~\ref{lsimresin}, $L_L$ is a selfadjoint operator and then, recurring
to Theorem~\ref{PopvsRel}, one obtains that $\zeta$ is
an isolated eigenvalue.

Let us compute the multiplicity of the eigenvalue $\zeta$. To this
end, observe that $\Nk\zeta L\subset \Nk\zeta{A^{*}}$. Also,
\begin{equation*}
\begin{split}
\dim [\ker(L-\zeta I)]&=\dim \Nk\zeta L\\&\leq\dim \Nk\zeta{A^{*}}=n\,.
\end{split}
\end{equation*}
\end{proof}
\begin{definition}
  \label{def:regular}
A relation $T$ is said to be regular if its quasi-regular
set is the whole complex plane, that is $\hat\rho(T)=\C$.
\end{definition}
\begin{corollary}\label{lainesp0}
  Let $A$ be a closed, regular, symmetric relation with
  $\eta_-(A)=n$. Assume that $L$ is a maximal dissipative
  extension of $A$.
  \begin{enumerate}[(i)]
  \item If $L$ is selfadjoint, then its spectrum consists of
    isolated eigenvalues of multiplicity at most $n$.
  \item If $L$ is not selfadjoint, then its spectral core consists of
    eigenvalues of multiplicity at most $n$.
  \item If $n=1$, then every number in $\C_{+}\cup\R$ is an
    eigenvalue of one, and only one, realization of $L$.
  \end{enumerate}
\end{corollary}
\begin{proof}
  First note that, due to the regularity of $A$, its deficiency
  indices are equal. Now, (i) is a consequence of
  Proposition~\ref{lainesp} since $\R\subset\hat\rho(A)$. For proving
  (ii), consider $\zeta\not\in\hat\rho(L)$ and repeat the
  argumentation of the proof of Proposition~\ref{lainesp} to show that
  $\ran(L-\zeta I)$ is closed so $\zeta\not\in\sigma_c(L)$. The
  multiplicity of any eigenvalue is computed as in the proof of
  Proposition~\ref{lainesp}. (iii) follows from
  Propositions~\ref{mdrofsrn} and \ref{mdrofsrn0}.
\end{proof}
\begin{remark}
  \label{rem:spectrum-discrete}
  Corollary~\ref{lainesp0} admits a refinement: the spectrum of the
  maximal dissipative extension $L$ is discrete. A way to prove it, in
  the case of deficiency indices $(1,1)$, is based on the fact that
  any closed regular symmetric operator with indices $(1,1)$ is
  unitarily equivalent to the multiplication operator in a dB space
  (see Subsection~\ref{sec:db-spaces} and, particularly, Theorem~\ref{spomvi}).
\end{remark}

\section{Applications to nondensely defined operators}
\label{sec:applications}

\subsection{Jacobi operators}
Consider the Hilbert space $l_2(\N)$, i.\,e. the space of
square-summable sequences. Fix two real sequences
$\{b_k\}_{k=1}^\infty$ and $\{q_k\}_{k=1}^\infty$ such that $b_k>0$
for $k\in\nats$ and let $J$ be the operator whose matrix
representation with respect to the canonical basis
$\{\delta_k\}_{k=1}^\infty$ of $l_2(\N)$ is
\begin{equation}
\label{jacob}
  \begin{pmatrix}
    q_1&b_1&0&0&\cdots\\
    b_1&q_2&b_2&0&\cdots\\
    0&b_2&q_3&b_3&\cdots\\
    0&0&b_3&q_4&\cdots\\
    \vdots&\vdots&\vdots&\vdots&\ddots
  \end{pmatrix}\,;
\end{equation}
see \cite[Sec.\,47]{MR1255973} for the definition of a matrix
representation for an unbounded closed symmetric operator.

Consider the difference equation
\begin{equation}
\label{kerA}
 b_{k-1}\phi_{k-1}+q_k\phi_k+b_k\phi_{k+1}=\zeta \phi_k,\,\,\, \zeta\in\C\,,
\end{equation}
for $k\in\N$ with $b_0=0$.
Setting $\phi_1=1$, one solves recurrently \eqref{kerA} and  $\phi_k$
is a polynomial of degree $k-1$ in $\zeta$, denoted here by $\pi_k(z)$,
and known as the $k-1$-th
polynomial of the first kind associated to \eqref{jacob}. Similarly,
$\phi_k$ is a polynomial of degree $k-2$ if one sets $\phi_1=0$ and
$\phi_2=1/b_1,$ in \eqref{kerA}. In this case $\phi_k$ is the $k-1$-th
polynomial of the second kind associated to \eqref{jacob} and
denoted by $\theta_k(z)$. It holds true that
\begin{equation}
  \label{eq:relation-polynomial-first-second}
  \theta_{k+1}(z)=b_1^{-1}\widetilde{\pi}_k(z)\,,
\end{equation}
where
$\widetilde{\pi}_k(z)$ is the polynomial obtained from $\pi_k(z)$
substituting $b_j,q_j$ by $b_{j+1},q_{j+1}$ ($j=1,\dots,k-1$).
\begin{remark}
  \label{rem:0-0-1-1}
The symmetric operator $J$ has deficiency indices $(0,0)$ or
$(1,1)$. The first case is characterized by the divergence of the
series $\sum_k|\pi_k(\zeta)|^2$ for all $\zeta\in\C\backslash\R,$
while the second case by the convergence of it (see
\cite[Chap.\,IV]{MR0184042} and
\cite[Chap.\,VII]{MR0222718}).
\end{remark}
Suppose that $J$ is selfadjoint and consider the linear operator
\beao B=J\rE{\oM{\dom J}{\Span \{\delta_1\}}}\,.  \eeao The operator
$B$ is closed, non-densely defined, and symmetric. By
\eqref{fodiidrsy}, $B$ has indices $(1,1)$.

\begin{proposition}
\label{prop:dissipative-ext-matrices}
The maximal dissipative extensions of $B$ are in one-to-one correspondence
with $\tau\in\C_{+}\cup\R\cup\{\infty\}$ and they are perturbations of
$J$ given by \bea\label{exoBop} J(\tau)=\llb \vE f{g+\tau\ip
  {\delta_{1}} f\delta_{1}} \ :\ \vE fg\in J\rrb,\,\,\,
\tau\neq\infty, \eea and \bea\label{exoBnop} J(\infty)=B\dotplus
\Span\llb\vE{0}{\delta_1}\rrb\,,  \eea
  where, for
  $\tau\in \R\cup\{\infty\}$, $J({\tau})$ is
  selfadjoint.
Furthermore,
\bea\label{taotB} B^*=J\dotplus
\Span\llb\vE{0}{\delta_1}\rrb.
\eea
\end{proposition}
\begin{proof}
  Fix $\zeta\in \C_{+}$, then $\pi(\zeta)$ and $\theta(\zeta)$ do not
  belong to $l_2(\N)$ in view of Remark~\ref{rem:0-0-1-1} and
  \eqref{eq:relation-polynomial-first-second}.
  According to \cite[Chap.1
  Sec.\,3]{MR0184042} (see also\cite[Chap. VII]{MR0222718}),
  there exists a unique function $m(\cdot):\C\setminus\R\rightarrow\C$
  satisfying $m(\overline\zeta)=\overline{m(\zeta)}$ and
  $(\im \zeta)(\im m(\zeta))>0$ such
  that $$\psi(\zeta)=\theta(\zeta)+m(\zeta)\pi(\zeta)\in\dom J\,.$$
  By a straightforward computation, one obtains \bea\label{psiz}
  J\psi(\zeta)=\delta_{1}+\zeta\psi(\zeta)\eea so that, for
  every $f\in \dom B$, one has
 \begin{align*}
 \ip{ f}{\zeta\psi(\zeta)}&=\ip{ f}{\delta_{1}+\zeta\psi(\zeta)}= \ip{f}{J\psi(\zeta)}\\
 &= \ip{J f}{\psi(\zeta)}= \ip{Bf}{\psi(\zeta)}\,,
 \end{align*}
 which means that $\vE{\psi(\zeta)}{\zeta\psi(\zeta)}\in B^{*}$. Therefore
 \bea\label{dsoaB}
 \Nk\zeta{B^{*}}=\Span\llb\vE{\psi(\zeta)}{\zeta\psi(\zeta)}
 \rrb\,\,\mbox{ and }\,\, \Nk{\overline
   \zeta}{B^{*}}=\Span\llb\vE{\psi(\overline{\zeta})}{\overline{\zeta}\psi(\overline{\zeta})}
 \rrb.  \eea
since $\eta_{\pm}(B)=1$.
 By Theorem
 \ref{Fovon00}, it holds that \bea\label{eq:bstbps} B^*=B\dotplus
 \Span\llb\vE{\psi(\overline{\zeta})}{\overline{\zeta}\psi(\overline{\zeta})}
 \rrb \dotplus \Span\llb\vE{\psi(\zeta)}{\zeta\psi(\zeta)} \rrb. \eea

 Now, since $\dom B$ as well as $\psi(\zeta) $ and
 $\psi(\overline\zeta)$ belong to $\dom J$, it follows that
 $\dom B^{*}\subset \dom J$. Hence $J$ and $B^{*}$ have the same
 domain. Observe that $J$ and $\oP{\{0\}}{\Span\{\delta_{1}\}}$ are
 linearly independent so that $J\dotplus Z\subset B^{*}$. On the other
 hand, \eqref{eq:bstbps} implies that
there exist $f\in \dom B$ and $a,\,b\in \C$ such that
$$\vE hk=\vE{ f+a\psi(\zeta)+b\psi(\overline\zeta)}{B f+a\zeta\psi(\zeta)+b\overline\zeta\psi(\overline\zeta)}$$
for every $\vE hk\in B^{*}$.
Then
by \eqref{psiz}
\begin{align*}
\vE hk&=\vE{f+a\psi(\zeta)+b\psi(\overline\zeta)}{J f+a(\delta_{1}+\zeta\psi(\zeta))+b(\delta_{1}+\overline\zeta\psi(\overline\zeta))}+\vE{0}{-(a+b)\delta_{1}}\\
&=\vE{f+a\psi(\zeta)+b\psi(\overline\zeta)}{J(f+a\psi(\zeta)+b\psi(\overline\zeta))}+\vE{0}{-(a+b)\delta_{1}}\in J\dotplus Z.
\end{align*}
We have proven \ref{taotB}. Now we turn to the proof of \eqref{exoBop}
and \eqref{exoBnop}. Note that these equations yield
closed dissipative extensions of $B$, which are therefore maximal.
Theorem \ref{Fovon01d} asserts that every maximal dissipative
extension $J(\beta)$ of $B$ is given by \bea\label{FvonB}
J(\beta)=B\dotplus (\mF V_\beta-\mF I)\Nk \zeta {B^*}, \eea with
$\mF V_\beta:\Nk \zeta {B^*}\rightarrow \Nk {\overline{\zeta}} {B^*}$
being a closed contraction. On the basis of \eqref{dsoaB}, one
concludes that all the
contraction mappings are in one-to-one correspondence with
$\beta\in\tE\cup\tE_{i}$ (i.\,e. $\abs{\beta}\le 1$) given by
$$\mF V_\beta\left(\vE {\psi(\zeta)}{\zeta\psi(\zeta)}\right)=
\beta\vE{\psi(\overline{\zeta})}{\overline{\zeta}\psi(\overline{\zeta})}\,,$$
whence, by means of \eqref{FvonB}, one arrives at
\bea\label{Atvon02}
 J(\beta)=B\dotplus \Span \llb \vE {\beta\psi(\overline{\zeta})-\psi(\zeta)}
 {\overline{\zeta}\beta\psi(\overline{\zeta})-\zeta\psi(\zeta)}\rrb.
 \eea
The last equality implies that $\dom J(\beta)\subset \dom J$. Take the M\"obius transformation
$$\beta_{\tau}=\frac{1+\tau m(\zeta)}{1+\tau m(\overline \zeta)}.$$
Since $m(\zeta)\in\C_{+}$, one has that $\beta_{\tau}$ maps
$\C_{+}\cup\R\cup\{\infty\}$ onto $\tE\cup\tE_{i},$ with
$\beta_{\infty}=m(\zeta)/m(\overline \zeta).$ Then for
$\tau\neq\infty$ it follows from \eqref{Atvon02} that, for any
$\vE hk\in J(\beta_{\tau}),$ there exists $f\in \dom B$ such that
\bea\label{cosaeB}
\vE hk&=&\vE{ f+\alpha(\beta_{\tau}\psi(\overline\zeta)-\psi(\zeta))}{Bf+\alpha(\beta_{\tau}\overline\zeta\psi(\overline\zeta)-\zeta\psi(\zeta))}\nonumber \\
&=&\vE{ f+\alpha(\beta_{\tau}\psi(\overline\zeta)-\psi(\zeta))}{J f+\alpha[\beta_{\tau}(\delta_{1}+\overline\zeta\psi(\overline\zeta))-(\delta_{1}+\zeta\psi(\zeta))]+\alpha(1-\beta_{\tau})\delta_{1}}\nonumber\\
&=&\vE{f+\alpha(\beta_{\tau}\psi(\overline\zeta)-\psi(\zeta))}{J
  [f+\alpha(\beta_{\tau}\psi(\overline\zeta)-\psi(\zeta))]+\alpha(1-\beta_{\tau})\delta_{1}}.
\eea
Note that $\psi_{1}(\zeta)=m(\zeta)$, so
\begin{align*}
\tau\ip{\delta_{1}}h&=\tau\ip{\delta_{1}}{ f+\alpha(\beta_{\tau}\psi(\overline\zeta)-\psi(\zeta))}\\
&=\alpha\tau(\beta_{\tau}m(\overline\zeta)-m(\zeta))\\
&=\alpha\tau\left(\frac{1+\tau m(\zeta)}{1+\tau m(\overline \zeta)}m(\overline\zeta)-m(\zeta)\right)\\
&=\alpha\left(1-\frac{1+\tau m(\zeta)}{1+\tau m(\overline \zeta)}\right)=\alpha(1-\beta_{\tau}).
\end{align*}
Thus \eqref{cosaeB} yields
$\vE hk=\vE h{Jh+\tau\ip{\delta_{1}}h\delta_{1}}\subset J(\tau)$. Due
to maximality it follows that $J(\beta_{\tau})=J(\tau)$.

For $\beta_{\infty}=m(\zeta)/m(\overline \zeta).$ the expression
$\beta_{\infty}\psi_{1}(\overline\zeta)-\psi_{1}(\zeta)$ vanishes. Thus,
by \eqref{Atvon02},  it follows that  $\dom
J(\beta_{\infty})\subset\dom B.$ Thence, according to \eqref{cosaeB}, for every  $\vE hk\in J(\beta_{\infty}),$
$$\vE hk=\vE h{Bh}+\vE0{\alpha(1-\beta_{\tau})\delta_{1}}\in B\dotplus
Z=J(\infty).$$ Therefore $J(\beta_{\infty})=J(\infty).$
\end{proof}
From what has been said, all the maximal dissipative extensions
\eqref{exoBop} of $B$ have the representation
  \beao
J(\tau)=\begin{pmatrix}
 q_1+\tau&b_1&0&0&\cdots\\
  b_1&q_2&b_2&0&\cdots\\
  0&b_2&q_3&b_3&\cdots\\
  0&0&b_3&q_4&\cdots\\
  \vdots&\vdots&\vdots&\vdots&\ddots
 \end{pmatrix}.
 \eeao

\subsection{Operator of multiplication in dB spaces}
\label{sec:db-spaces}
There are two essentially different ways of defining a de Branges space
(dB space) \cite[Chap.\,2]{MR0229011}.
The following one has an axiomatic structure:

A nontrivial Hilbert space of entire functions $\cA B$ is said to be a dB space when for every function $f(z)$ in the space, the following holds:
\begin{enumerate}[({A}1)]
\item\label{axidb1} For every $w\in \C\backslash\R,$ the linear functional $f(\cdot)\ \mapsto f(w)$ is continuous;
\item\label{axidb2} for every non-real zero $w$ of $f(z),$ the function $f(z)(z-\overline{w})(z-w)^{-1}$ belongs to $\cA B$ and has the same norm as $f(z)$;
\item\label{axidb3}  the function $\gA fz=\overline{f(\overline{z})}$ also belongs to $\cA B$ and has the same norm as $f(z).$\end{enumerate}

Due to the polarization identity, (A\ref{axidb3}) implies
\begin{align}\label{eqipidb}
\ip {f(\zeta)}{g(\zeta)}=\ip {\gA g \zeta}{\gA f \zeta}
\end{align}
for every $f(z),\ g(z)\in \cA B$.

By the Riesz lemma, (A\ref{axidb1}) is equivalent to the existence of
a unique reproducing kernel $k(z,w)$ that belongs to $\cA B$ for every
$w\in \C\backslash\R$ and satisfies \bea\label{porkw} \ip
{k(\zeta,w)}{f(\zeta)}=f(w),\eea for every $f(z)\in \cA B$. Besides,
$k(w,w)=\ip{k(\zeta,w)}{k(\zeta,w)}>0$ as a consequence of
(A\ref{axidb2}) (see the proof of \cite[Thm.\,23]{MR0229011}). Note
also that
$\overline{k(z,w)}=k(w,z)$.
Finally, in view of \eqref{eqipidb}, for every $f(z)\in \cA B$,
\begin{equation*}
\ip{\gA k{\zeta,w}}{f(\zeta)}=\overline{\ip {k(\zeta,w)}{\gA f\zeta}}
=\overline{\gA fw}=\ip{k(\zeta,\overline w)}{f(\zeta)}\,,
\end{equation*}
whence $\overline{k(\overline z,w)}=k(z,\overline w)$.

The operator of multiplication by the independent variable in $\cA B$ is defined by the relation
\begin{equation}\label{omxvi}
S=\llb\vE {f(z)}{zf(z)}\ :\ f(z),\ zf(z)\in \cA B\rrb\,.
\end{equation}
Clearly it is an operator and \cite[Prop.\, 4.2, Cors.\,4.3 and
4.7]{MR1664343} show that $S$ is closed, regular, symmetric, with
deficiency indices $(1,1)$, and not necessarily densely defined.

Fix $w\in\C_{+}$. It follows from \eqref{porkw} that for every $\vE {f(z)}{(z-w)f(z)}\in S-wI$
$$\ip {k(\zeta,w)}{(\zeta-w)f(\zeta)}=(w-w)f(w)=0,$$
which implies that
$k(z,w)\in \ker(S^*-\overline wI)=\dom \Nk{\overline{w}}{S^*}$. Since
$\eta_{\pm}(S)=1$, one has
\bea\label{deoemvi} \Nk {\overline{w}}{S^*}=\Span\llb\vE
{k(z,w)}{\overline wk(z,w)}\rrb; \:\: \Nk w{S^*}=\Span\llb\vE
{k(z,\overline w)}{ wk(z,\overline w)}\rrb.  \eea

Equation \eqref{von01} now reads \bea\label{bstbps} S^*=S\dotplus
\Span\llb\vE {k(z,w)}{\overline wk(z,w)}\rrb \dotplus \Span\llb\vE
{k(z,\overline w)}{ wk(z,\overline w)}\rrb. \eea Furthermore by
Theorem \ref{Fovon01d}, every maximal dissipative extension $S_\tau$
of $S$ is given by
 \begin{align}\label{eomxvi1}
  S_\tau={S}\dotplus {(\mF V_\tau-I)\Nk w{S^*}}
 \end{align}
with $\mF V_\tau:\Nk w {S^*}\rightarrow \Nk {\overline{w}} {S^*}$
being a closed contraction given by
$$\mF V_\tau\left(\vE {k(z,\overline w)}{ wk(z,\overline w)}\right)=
\tau\vE {k(z,w)}{\overline wk(z,w)}\,,$$ where $\abs{\tau}\le 1$. Note
that the form of $\mF V_\tau$ has been deduced from
\eqref{deoemvi}. Whence, by \eqref{eomxvi1}, one has
\bea\label{eomxvi2} S_\tau=S\dotplus \Span \llb \vE {\tau
  k(z,w)-k(z,\overline w)} {\tau\overline{w}k(z,w)-wk(z,\overline
  w)}\rrb.  \eea Notice that for any $\tau\in\C$ such that
$\abs{\tau}\le 1$, $S_{\tau}$ has the spectral properties given in
Corollary~\ref{lainesp0}. Moreover, for $\abs{\tau}=1$,
$\mF V_{\tau}$ is isometric and, as a consequence of Theorem
\ref{Fovon01d}, $S_\tau$ is a selfadjoint extension of $S$.

The other definition of dB space requires the Hardy space
\begin{align*}
\H_{2}(\C_{+}):=\llb f(z) \mbox{ holomorphic in } \C_{+}\ :\ \sup_{y>0}\int_{\R}|f(x+iy)|^{2}<\infty\rrb
\end{align*}
as well as an Hermite-Biehler function, which is an entire function $e(z)$ satisfying
\beao
|e(z)|>|\gA e z|,\,\,\,  z\in \C_{+},\eeao
whence it follows that $e(z)$ is a function without zeros in the half-plane $\C_{+}$.

The dB space associated with an Hermite-Biehler function $e(z)$
\cite[Sec.\,2]{MR1943095} is the linear manifold
\begin{align*}
\cA B(e):=\llb f(z) \mbox{ entire } :\, \frac{f(z)}{e(z)},\ \frac{\gA f z}{e(z)}\in \H_{2}(\C_{+})\rrb,
\end{align*}
equipped with the inner product
\begin{align*}
\ip {f(t)}{g(t)}_{e}:=\int_{\R}\frac{\overline{f(t)}g(t)}{|e(t)|^{2}}dt.
\end{align*}
Without loss of generality, let us assume that $e(z)$ not
only has no zeros in $\C_{+}$, but also in $\R$.

In \cite[Sec.\,5]{MR1664343} (see also \cite[Sec.\,19]{MR0229011}), it
is shown that, for any $w\in \C$, the expression
\begin{align}\label{kroBe}
k(z,w)=\frac{\gA eze(\overline w)-e(z)\gA e{\overline w}}{2\pi i(z-\overline w)}
\end{align}
is the reproducing kernel of $\cA B(e)$. Moreover, since $e(z)$ does
not have zeros on $\C_{+}\cup\R$, $k(z,w)$ has no zeros in
$\C_{+}\cup\R$ for every $w\in \C_{+}$.

For a given dB space $\mathcal{B}$ with reproducing kernel $k(z,w)$,
if one defines \bea\label{tdfwak} e_{ w_0}(z):=\frac{\pi(z-\overline
  w_0)}{(\im w_0)k( w_0, w_0)}k(z, w_0)\,,\quad w_0\in\mathbb{C}_+, \eea then
$e_{ w_0}(z)$ is an Hermite-Biehler function \cite[Sec.\,2]{MR3002855}
and $\cA B=\cA B(e_{ w_0})$ isometrically
\cite[Thm.\,7]{zbMATH06526214}. The reproducing kernel of
$\cA B(e_{ w_0})$ can be computed using \eqref{kroBe}:
\begin{align*}
k_{ w_0}(z, w_0)&=\frac{\gA {e_{ w_0}}ze_{ w_0}(\overline w_0)-e_{ w_0}(z)\gA {e_{ w_0}}{\overline w_0}}{2\pi i(z-\overline w_0 )}
\\&=\frac{-e_{ w_0}(z)}{2\pi i(z-\overline w_0)}\left(\frac{\pi(\overline w_0- w_0)}{(\im  w_0) k( w_0, w_0)}k( w_0, w_0)\right)
=\frac{e_{ w_0}(z)}{z-\overline w_0}\,.
\end{align*}
Therefore
\begin{align}\label{HbfidB}
e_{ w_0}(z)=(z-\overline w_0)k_{ w_0}(z, w_0).
\end{align}

The set of associated functions $\Assoc\cA B$ of a dB space  $\cA B$
is given by
\begin{equation*}
  \Assoc\cA B=\cA B+z\cA B\,.
\end{equation*}
For a $\tau\in\C$, define \begin{align}\label{HbfidB0}
\varphi_{\tau}(z):=\tau e(z)-\gA ez\,.
\end{align}
These entire functions belong to $\Assoc\cA B$ and determine the
maximal dissipative extensions of the multiplication operator.

\begin{theorem}\label{spomvi}
  Fix $ w \in\C_{+}$ and consider the dB space $\cA B(e_{ w })$ with
  $e_{ w }(z)$ given in \eqref{tdfwak}. All the maximal dissipative
  extension of the operator of multiplication $S$ (see \eqref{omxvi})
  are in one-to-one correspondence with the set of entire functions
  $\varphi_{\tau}(z)$, $\abs{\tau}\le 1$. These maximal
  dissipative extensions are given by
  {\small\begin{align}\label{spomvi0}
      S_\tau=\llb\vE{h_{\alpha}(z)}{zh_{\alpha}(z)-\alpha\varphi_{\tau}(z)}\
      :\
  \begin{array}{c}
   h_{\alpha}(z)=f(z)+\alpha(\tau k_ w (z, w )-k_ w (z,\overline w )),\\
    f(z)\in \dom S
  \end{array}\rrb.
 \end{align}}
Moreover, $\sigma(S_{\tau})=\{\lambda\in\C_{+}\cup\R\ :\
\varphi_{\tau}(\lambda)=0\}$. The eigenfunction
corresponding to $\lambda\in\sigma(S_{\tau})$ is
$h^{(\tau)}(z)=\frac{\varphi_{\tau}(z)}{z-\lambda}$.
\end{theorem}
 \begin{proof}
   All dissipative extensions $S_{\tau}$ are given by
   \eqref{eomxvi2}. If $\vE{h(z)}{g(z)}\in S_\tau$, then there exist
   $\vE {f(z)}{zf(z)}\in S$ and $\alpha\in\C$ such that
 $$\vE{h(z)}{g(z)}=\vE{f(z)+\alpha(\tau k_ w (z, w )-k_ w (z,\overline w ))}{zf(z)+\alpha(\tau\overline w k_ w (z, w )-w k_ w (z,\overline w ))}\,.$$
It follows from \eqref{HbfidB} and \eqref{HbfidB0} that
\begin{align*}
g(z)
 &=zh(z)-\alpha[\tau(z-\overline w)k_ w (z, w )-(z-w)k_ w (z,\overline w )]\\
  &=zh(z)-\alpha[\tau e_{w}(z)-\gA {e_{w}}z]\\
 &=zh(z)-\alpha\varphi_{\tau}(z),
 \end{align*}
whence \eqref{spomvi0} follows.

By Corollary~\ref{lainesp0}, for every $\lambda\in\hat{\sigma}(S_{\tau})$
(which is a subset of $\C_{+}\cup\R$ by Theorem~\ref{codrih}),
$\Nk{\lambda}{S_{\tau}}$ is one-dimensional. Thus, in view of
\eqref{spomvi0},  $\vE {h_{\alpha}(z)}{\lambda
  h_{{\alpha}}(z)}\in\Nk{\lambda}{S_{\tau}}$ with
 $\alpha\neq0$ and
 \begin{align}\label{HbfidB1}
\lambda h_{\alpha}(z)=zh_{\alpha}(z)-\alpha\varphi_{\tau}(z),
 \end{align}
so $ \varphi_{\tau}(\lambda)=0$.

On the other hand, if $\lambda\notin\sigma(S_{\tau})$, then
$(S_{\tau}- \lambda I)^{-1}\in \mathcal{B}(\H)$. Thus $\vE {k_ w (z, w
  )}{g(z)}\in (S_{\tau}- \lambda I)^{-1}$ which implies
$\vE {g(z)}{k_ w (z, w )+ \lambda g(z)}\in S_{\tau}$. Using
again \eqref{spomvi0}, one arrives at
 \begin{align*}
 k_ w (z, w )+ \lambda g(z)=zg(z)-\alpha\varphi_{\tau}(z)\,.
 \end{align*}
Therefore, for $z=\lambda $, $ k_{w}(\lambda,w)=-\alpha\varphi_{\tau}(\lambda)$.
Since $k_ w (z, w )$ has no zeros in $\C_{+}\cup\R$, one concludes
that  $\lambda\notin  \{\lambda\in\C_{+}\cup\R\ :\ \varphi_{\tau}(\lambda)=0\}$.

We have proven that
\begin{equation*}
  \hat\sigma(S_{\tau})\subset\{\lambda\in\C_{+}\cup\R\ :\
\varphi_{\tau}(\lambda)=0\}\subset\sigma(S_{\tau})\,.
\end{equation*}
Since the zeros of a nonzero entire function is a discrete set, the
spectral core is a discrete set. This implies that
$\rho(S_\tau)=\hat\rho(S_\tau)$ because of the maximality of $S_\tau$.

The fact that $h^{(\tau)}(z)$ is an eigenfunction of $S_\tau$
corresponding to the eigenvalue $\lambda$ is a consequence of  \eqref{HbfidB1}.
  \end{proof}
\subsection*{Acknowledgements}
The authors thank the anonymous referee
whose pertinent comments led to an improved presentation of this work.
\small
\def\lfhook#1{\setbox0=\hbox{#1}{\ooalign{\hidewidth
  \lower1.5ex\hbox{'}\hidewidth\crcr\unhbox0}}} \def\cprime{$'$}

\normalsize
\end{document}